\documentclass[12pt]{article}
\usepackage{CJKutf8}
\usepackage{times,amssymb,amsmath,amsfonts,eurosym,geometry,ulem,graphicx,caption,color,setspace,sectsty,comment,footmisc,caption,pdflscape,array,threeparttable,tikz,caption,natbib,bbm} 
\PassOptionsToPackage{draft}{graphicx}
\usepackage{subfig}
\newsavebox{\measurebox}
\usepackage[scale=2]{ccicons}
\usepackage{pgfplots}
\usepgfplotslibrary{dateplot}
\usetikzlibrary{intersections,shapes.multipart}
\usepackage{pifont}
\usetikzlibrary{tikzmark}
\usetikzlibrary{shapes,arrows,backgrounds,positioning}
\usetikzlibrary{intersections,shapes.multipart}
\usetikzlibrary{fit,calc}

\usetikzlibrary{intersections,shapes.multipart}
\usepackage{pifont}
\usetikzlibrary{tikzmark}
\usetikzlibrary{shapes,arrows,backgrounds,positioning}
\usetikzlibrary{intersections,shapes.multipart}
\usepackage{pifont}
\usepackage{tikz}
\usetikzlibrary{fit,calc}
\usepackage{booktabs} 

\usepackage[colorlinks=true,
citecolor=blue,
linkcolor=blue,
anchorcolor=blue,
urlcolor=blue]{hyperref}

\usepackage{tikz,pgfplots}
\usepackage{amsmath,amsthm}
\newtheorem{hypothesis}{Implication}
\newtheorem{assumption}{Assumption}

\newtheorem{corollary}{Corollary}
\newtheorem{proposition}{Proposition}

\theoremstyle{definition}
\newtheorem{definition}{Definition}
\theoremstyle{example}

\newcommand{\1}{\mathbf{1}}

\usepackage[titletoc,toc,title,page,header]{appendix}
\usepackage{minitoc}

\noptcrule
\usepackage{caption}

\title{A Formal Theory of Survey Experiment Generalizability: Attention and Salience\thanks{We thank Donald Green, Melody Huang, John Marshall, Cyrus Samii, Tara Slough, and Anton Strezhnev
; Participants at Experiment Conference at New York University,  Center for Data Science of Renmin University of China, Polmeth and MPSA 2024. We also thank Raymond Liu for his excellent research assistance.}}

\author{Jiawei Fu\footnote{Assistant Professor, Duke University \url{jiawei.fu@duke.edu}}  \quad Xiaojun Li\footnote{Professor of Political Science, University of British Columbia \url{xiaojun.li@ubc.ca}}}
\date{\today}

\begin{document}

\maketitle
\begin{abstract}

\noindent Survey experiments are widely used to identify causal effects in political science and the social sciences. Yet researchers are typically interested in more than the internal validity of an experimentally induced contrast. They also want to know whether the estimated effect corresponds to the effect in the real world. We develop a formal theory of survey experiment generalizability grounded in behavioral microfoundations. The theory highlights two mechanisms.  First, the survey environment shapes \emph{attention}: it determines which considerations enter the respondent's active consideration set. Second, it shapes \emph{salience}: conditional on consideration, it influences the relative weight assigned to those considerations. This framework yields two main results. Consideration-set compression generates amplification: survey-experimental effects can be larger in magnitude than their real-world counterparts, even for the same individuals, treatment content, and outcome. Context-dependent salience generates sign instability: the direction of the survey effect need not coincide with the direction of the corresponding real-world effect. The theory clarifies what survey experiments identify, when those effects are likely to generalize, and how survey designs can be modified to improve decision-environment transportability.
\vspace{.1in}


\end{abstract}

\thispagestyle{empty}
\doublespacing
\clearpage

\doparttoc 
\faketableofcontents 

\setcounter{page}{1}

\section*{Introduction}

Survey experiments are a core empirical tool in political science and the broader social sciences. Researchers use them to study candidate choice, policy support, accountability, punishment, among other outcomes \citep{eggers2018corruption,carnes2016voters,sanbonmatsu2002gender}. Because treatment is randomized, survey experiments are widely regarded as providing credible causal evidence. Between 2021 and 2025, a total of 204 survey experiments were published in leading political science journals, as shown in Table \ref{tab:journal}. Following the classification proposed by EGAP, we categorize these studies into five types: conjoint, endorsement, list, priming, and vignette/factorial.\footnote{https://egap.org/resource/10-things-to-know-about-survey-experiments/}
 Figure \ref{fig:dis} presents the distribution of these categories over time. Among them, vignette/factorial and conjoint designs are the two dominant approaches.

However, in most substantive applications, internal validity is not the end of the story. Researchers are often interested in a stronger claim: how the same information, attribute, or intervention would operate in real-world environments beyond the survey setting. For example, a candidate trait that shifts support in a conjoint experiment is often interpreted as evidence of how voters would weigh that trait in actual elections. Similarly, a corruption message that affects respondents’ evaluations in a vignette is often taken to reflect how such information would influence real-world political judgments. Survey experiments are often viewed as a powerful tool for learning about real-world effects \citep{gaines2007logic}. Yet, under what conditions is such generalizability warranted? More fundamentally, what mechanisms determine whether findings from survey experiments successfully translate to real-world behavior, and when might they fail to do so?

\begin{figure}[!ht]
    \centering
    \includegraphics[width=0.7\linewidth]{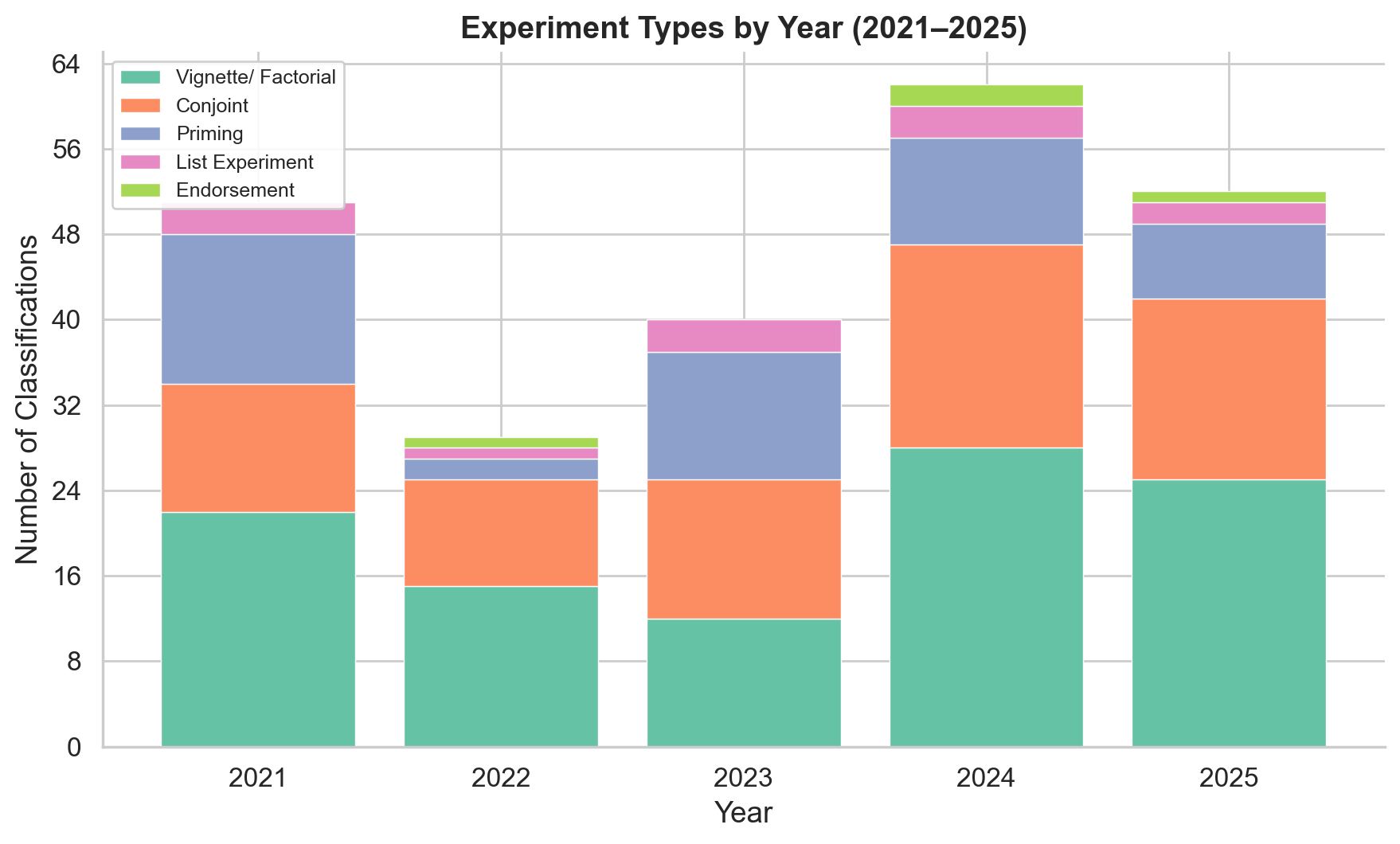}
    \caption{Distribution of Survey Experiments in AJPS, APSR, JOP, PA}
    \label{fig:dis}
\end{figure}

\begin{table}[htbp]
\centering
\small
\caption{Survey experiments by journal and subtype from 2012 to  2015}
\label{tab:journal}
\begin{tabular}{lrrrrrr}
\toprule
Journal & Total & Conjoint & Endorsement & List & Priming & Vignette/Factorial \\
\midrule
American Journal of Political Science & 41 & 14 & 2 & 1 & 8 & 20 \\
American Political Science Review & 55 & 16 & 0 & 6 & 12 & 27 \\
Political Analysis & 9 & 8 & 0 & 1 & 1 & 0 \\
The Journal of Politics & 99 & 33 & 2 & 4 & 24 & 55 \\
\midrule
Total & 204 & 71 & 4 & 12 & 45 & 102 \\
\bottomrule
\end{tabular}
\vspace{0.35em}

\parbox{0.92\linewidth}{\footnotesize Note: Subtype counts are not mutually exclusive because some papers have more than one survey experiment subtypes.}
\end{table}

In this study, we develop a formal framework that integrates a behavioral model with the potential outcomes framework to characterize two mechanisms that shape the generalizability of survey experiments. In the existing literature, generalizability and external validity are typically understood as the ability to make inferences about a broader population \citep{findley2021external}. By contrast, we take a more fundamental step by asking: even when we fix the same individual, the same treatment, and the same outcome measure, does the effect estimated in a survey experiment inform the corresponding real-world effect? If the answer is negative, then concerns about generalizing survey experimental results to other populations become secondary.

More specifically, following \citet{egami2023elements}, we focus on two dimensions of generalizability: effect magnitude and effect sign. \textit{Effect magnitude generalization} requires that the size of the estimated effect in the survey experiment corresponds to its magnitude in the real world. \textit{Effect sign generalization}, by contrast, requires only that the direction (positive or negative) of the effect is preserved across settings. Importantly, an experiment need not satisfy both criteria simultaneously; the relevant notion of generalizability depends on the researcher’s objective. For example, policy evaluations often require accurate estimates of effect magnitudes to assess welfare implications, and thus demand magnitude consistency. In contrast, when experiments are used to test predictions derived from formal theory, preserving the direction of the effect is often sufficient. In some cases, experimental designs may even deliberately amplify treatment variation to increase statistical power, making sign consistency the more appropriate benchmark.

Building on our model, we characterize two mechanisms that govern the generalizability of both effect magnitude and effect sign. These mechanisms rest on two behavioral microfoundations.

The first is \textit{limited attention}. Conventional models of decision-making typically assume that individuals possess unlimited attention—that is, they account for all relevant attributes when making choices. While such assumptions yield tractable and often sharp theoretical predictions, they have been increasingly challenged by empirical anomalies. Classic examples include the paradoxes documented by \citet{allais1953comportement} and \citet{ellsberg1961risk}, which illustrate systematic violations of expected utility theory.\footnote{The former shows that individuals frequently violate the independence axiom of expected utility theory, while the latter demonstrates ambiguity aversion inconsistent with subjective expected utility theory.} In practice, individuals rarely process all potentially relevant considerations. Instead, they rely on a selective and context-dependent \textit{consideration set}. This notion has long been recognized in the social sciences. For example, consumer choices may favor option $y$ in the presence of $x$, yet this preference can reverse when $x$ is not readily available \citep{masatlioglu2016revealed}. We incorporate this idea into the causal inference framework for survey experiments and show that differences in consideration sets across environments play a central role in shaping generalizability. In particular, the more restricted consideration sets typically induced by survey environments can mechanically amplify estimated effects relative to their real-world counterparts.

The second mechanism arises from \textit{salience}. Even within a given consideration set, attributes are not weighted equally; some dimensions attract disproportionate attention. This phenomenon is well documented. For instance, \citet{chetty2009salience} show that consumers underreact to taxes when tax-inclusive prices are not prominently displayed, but reduce demand by approximately eight percent when such prices are made salient. To formalize this mechanism, we draw on the psychologically grounded model of choice developed by \citet{bordalo2012salience,bordalo2013salience}, in which decision weights are distorted by the relative salience of attributes. We show that salience, when combined with limited attention, can generate failures of effect sign generalization. The key intuition is that salience distorts the relative weighting of attributes, effectively inducing correlations among them in the decision process. When consideration sets differ between survey and real-world environments, these salience-driven distortions operate on different subsets of attributes, potentially reversing the direction of the estimated effect. 

Although the core components of our theory—consideration sets and salience—are not directly observable, we assess the empirical relevance of our framework by deriving a set of testable implications from the theoretical results. A central prediction is that the estimated effect attenuates, and may even reverse, as the number of attributes increases in some experiments. Intuitively, expanding the attribute space alters both the consideration set and the relative salience of each attribute, thereby changing the mapping from experimental effects to their real-world counterparts. We examine this prediction using data from multiple existing studies and supplement the analysis with additional pre-registered conjoint experiments.\footnote{The study is pre-registered and approved by the authors’ Institutional Review Board; further details are provided in the SI \ref{si:design}.} As the well-known adage suggests, “all models are wrong, but some are useful.” We do not claim that our framework exhausts all mechanisms underlying the empirical patterns we document. Rather, our goal is to isolate two theoretically grounded mechanisms—limited attention and salience—and to show that they generate distinctive, testable predictions. 

Attention and salience have already emerged as important considerations in the study of conjoint experiments. For example, \citet{jenke2021using} and \citet{bansak2025odd} employ eye-tracking techniques to examine how respondents allocate attention across attributes. We build on this line of work by providing a formal framework that characterizes the conditions under which such patterns arise and how they affect causal interpretation. We emphasize that the definitions of attention and salience are quite versatile in the literature. Here, we define limited attention as a situation in which only a subset of attributes is under consideration, and we define salience as the extent to which the value of an attribute affects its importance. These definitions follow the framework of \citet{bordalo2012salience}. This idea is closely related to the concept of extremity \citep{liu1998catastrophic, abelson2014attitude, krosnick1992case}.

More broadly, our results suggest that, among survey-based designs, well-constructed conjoint experiments may offer advantages for generalizability. By presenting multiple attributes simultaneously, conjoint designs encourage respondents to consider a richer set of factors, thereby mitigating distortions induced by limited attention and salience. When the set of attributes included in the experiment sufficiently captures the core and relevant dimensions of real-world decision-making, concerns about generalizability are correspondingly reduced. This perspective helps rationalize prior empirical findings documenting the robustness of conjoint experiments \citep{jenke2021using,hainmueller2015validating,bansak2018number, bansak2021beyond}, and provides a microfoundation for understanding when such robustness should be expected.

Our research intersects with and contributes to several important strands of literature. At its core, our study engages with a central question in the literature on external validity (\citealt{barabas2010survey}; \citealt{list2005laboratory}; \citealt{gaines2007logic}), particularly the external validity of survey experiments. \citet{slough2023external,slough2022sign} develop a formal framework to clarify external validity by identifying conditions under which different studies are target-equivalent and target-congruent. Their framework emphasizes the importance of harmonizing studies along key dimensions, including treatment, control, and outcome measurement. Similarly, \citet{egami2023elements} propose a unified framework encompassing multiple dimensions of external validity—population, context, treatment, and outcome—building on the typology of \citet{cook2002experimental}. They further develop estimators for both effect magnitude and sign generalization. We complement this literature by providing a behavioral foundation for understanding when and why effect magnitude and sign generalization succeed or fail.

Much of this literature focuses on generalizing results to new populations (\citealt{mullinix2015generalizability}; \citealt{huang2022sensitivity}). By contrast, our paper examines a more fundamental form of external validity: the consistency between experimental effects and real-world effects, even when holding fixed the same individuals, treatments, and outcomes. A variety of factors may undermine such consistency, including well-known phenomena such as the Hawthorne effect \citep{adair1984hawthorne}. We contribute to this literature by providing a new explanation tailored to settings involving multidimensional decision-making.

More specifically, concerns about the generalizability of survey experiments have been widely documented, with mixed empirical evidence \citep{barabas2010survey,findley2017external,hainmueller2014survey}. Existing explanations for threats to external validity include features of experimental design \citep{de2020commensurability} and information (non)equivalence across settings \citep{dafoe2018information}. Departing from purely statistical approaches, we develop a behavioral framework that highlights the roles of limited attention and salience as key mechanisms underlying these discrepancies.

Our work also contributes to the emerging literature on the Theoretical Implications of Empirical Methods (TIEM) \citep{FU_SLOUGH_2026, slough2023external,de2020commensurability,slough2023phantom}, which evaluates empirical methodologies through the lens of formal theory. In contrast to the Empirical Implications of Theoretical Models (EITM) approach, TIEM emphasizes how methodological choices themselves embed implicit theoretical assumptions. Within this perspective, our application of decision theory and behavioral economics to survey experiments reveals systematic discrepancies between experimental estimates and real-world effects, extending beyond the aggregate-level inconsistencies documented by \citet{abramson2019we}.

Finally, we contribute to the literature on salience in political economy. While salience theory has been widely applied to topics such as electoral behavior, party competition, agenda setting, and judicial decision-making \citep{moniz2020issue, dragu2016agenda,riker1986art,ascencio2015endogenous,guthrie2000inside,viscusi2001jurors, bartels1986issue,niemi1985new,repass1929issue}, its implications for research design have received less attention. By incorporating the framework developed by \citet{bordalo2012salience,bordalo2013salience,bordalo2016competition} into the study of survey experiments, we highlight salience as a central consideration in experimental methodology, thereby extending its scope and applicability.

\section{Model of Survey Experiments}

Causal effects are typically formalized within the potential outcomes framework. To illustrate the mechanisms underlying generalizability, we develop a model that integrates behavioral microfoundations of attention and salience with the potential outcomes framework. We explicitly incorporate environmental features into the potential outcomes and emphasize that the target estimand in a survey experiment depends on the consideration set and the salience of attributes.

\subsection{Attention and Salience}

Consider a survey experiment conducted in a survey environment. Let \(E^s\) denote the survey environment, and \(E^r\) denote the target real-world environment. Let the informational content available to respondent \(i\) be summarized by an attribute vector $X_i = (X_{i1},\dots,X_{iK}) \in \mathcal{X}$.
The components of \(X_i\) can be interpreted broadly. In a conjoint experiment, they correspond to profile attributes. In a vignette experiment, they represent features of a scenario. In an informational treatment, they capture the arguments, facts, or cues embedded in the message. In candidate or policy evaluation tasks, they describe the characteristics of the object being evaluated. Throughout the paper, we refer to $X_{ik}$ as attribute, and to its specific realization $x_{ik}$ as the value of the attribute.

Individuals have limited attention. A large body of research shows that, when making decisions—such as purchasing a computer—individuals do not evaluate the full set of available options or attributes \citep{hausman2008mindless}. Instead, they focus on a subset of relevant information. In the literature, the subset of attributes or alternatives that a respondent attends to when making a decision is referred to as the consideration set \citep{hausman2008mindless,wright1977phased}. It is well established that, due to cognitive constraints, individuals cannot allocate sufficient attention to all potentially relevant attributes \citep{stigler1961economics,jones2005politics,chetty2009salience}. As a result, the consideration set is typically a strict subset of the full attribute space that may, in principle, influence a decision. Formally, let
\[
C_i(E) \subseteq \{1,\dots,K\}
\]
denote respondent \(i\)'s active consideration set under environment \(E \in \{E^r,E^s\}\). If \(k \notin C_i(E)\), then attribute \(k\) does not enter the active evaluation in environment \(E\). 

Following the salience model \citep{bordalo2012salience,bordalo2013salience}, given an attribute vector $X_i$, respondent $i$ evaluates it in environment $E$ as:
$$
V_i(X_i,E)=\sum_{k=1}^K \alpha_{ik}(X_i,E) u_{ik}(X_{ik})=\sum_{k \in C_i(E)} \alpha_{ik}(X_i,E) u_{ik}(X_{ik})
$$ where $\alpha_{ik}(X_i,E)>0$ denotes the salience weight assigned to attribute $k$, satisfying $\sum_{k=1}^K \alpha_{ik}(X_i,E)=1$. The function $u_{ik}(X_{ik})$ represents the utility contribution of attribute $k$. We assume that the baseline utility component $u_{ik}$ is invariant across environments. This restriction allows us to isolate the role of attention and salience—captured by $C_i(E)$ and $\alpha_{ik}(X_i,E)$—in shaping differences between the survey and real-world evaluations. It is worth noting that the environment $E$ may influence many aspects of the decision-making process. In this study, however, we focus exclusively on the roles of attention and salience.

The key idea is that salience $\alpha_{ik}(X_i,E)$ need not be fixed. As emphasized by \citet{taylor1982stalking}, "salience refers to the phenomenon that when one’s attention is differentially directed to one portion on the environment rather than to others, the information contained in that portion will receive disproportionate weighing in subsequent judgment." This idea is formalized by a salience function $\sigma(x_k,\overline{X}_k)$, where $x_k$ denotes the realized value of attribute $X_k$, and $\overline{X}_k$ is a reference point. The reference may be individual-specific or determined by the experimental context. For example, in a vignette experiment where the attribute is the level of corruption, each respondent $i$ may have a baseline perception of the typical level of corruption among politicians. In a conjoint experiment, where respondents evaluate two hypothetical politicians, a natural reference point is the average level of corruption across the two profiles. 

The salience function captures the “distance” between an attribute value and its reference point. Intuitively, attributes that deviate more from the reference point receive greater salience. A more formal treatment is provided in \citet{bordalo2012salience}. The salience function is typically assumed to satisfy several axiomatic properties. First, for any two intervals $[x,y]$ and $[x',y']$, with $[x,y]$ fully contained within $[x',y']$, it holds that $\sigma(x,y) < \sigma(x',y')$. This reflects the psychological principle that larger contrasts are more perceptually salient. Second, the function satisfies homogeneity of degree zero: $\sigma(b x, b y) = \sigma(x,y)$ for any positive scalar $b$. A commonly used functional form that satisfies these properties is $\sigma(x,y) = \frac{|x-y|}{y}$.

Now we introduce how salience operates. Let $r_{ik}(X_i,E)\in\{1,\dots,|C_i(E)|\}$
denote the salience rank of attribute \(k\) among the active considerations in environment \(E\), where smaller values indicate greater salience. The effective salience weight $\alpha_{ik}$ takes the form
\begin{equation}
\alpha_{ik}(X_i,E)
=
\frac{\bar{\alpha}_{ik}\,\delta^{\,r_{ik}(X_i,E)-1}}
{\sum_{j\in C_i(E)} \bar{\alpha}_{ij}\,\delta^{\,r_{ij}(X_i,E)-1}},
\qquad k\in C_i(E),
\label{eq:salience_weight}
\end{equation}
where \(\delta \in (0,1]\) indexes the strength of salience, and $\bar{\alpha}_{ik}$ denotes the baseline salience. When attribute $j$ is more salient than attribute $k$, as determined by the salience function, i.e., $\sigma(x_{ij},\overline{X}_{ij})>\sigma(x_{ik},\overline{X}_{ik})$, it follows that $r_{ij}<r_{ik}$. Consequently, under the effective salience weights, attribute $X_{k}$ is more heavily discounted, since $\delta^{r_{ik}-1}<\delta^{r_{ij}-1}$. 

To illustrate the mechanism, consider a simple example. For ease of exposition, we suppress the individual subscript $i$. Suppose a candidate-choice conjoint experiment with two attributes: age $X_1$ and gender $X_2$. Let $\alpha_1$ and $\alpha_2$ denote the baseline salience weights, reflecting the importance of age and gender in candidate evaluation. When respondents observe two hypothetical candidates, salience is determined by how each attribute value deviates from its reference point. Without loss of generality, suppose that in the conjoint setting the reference for each attribute is given by the average across the two candidates $\overline{X}_j=\frac{x_1+x_2}{2}$. If, for a given comparison, gender is more salient than age, i.e., $\sigma(x_1,\overline{X}_1) < \sigma(x_2,\overline{X}_2)$, then $r_1 = 2$ and $r_2 = 1$. The resulting evaluation $V(X_1,X_2,E^s)$ with effective salience is

\begin{equation}\label{equ:sal}
    V(X_1,X_2,E^s) = \begin{cases}
     \frac{\alpha_1}{\alpha_1 + \delta \alpha_2} u_1(x_{1}) + \frac{\delta \alpha_2}{\alpha_1 + \delta \alpha_2} u_2(x_{2})  & \text{if }  \sigma(x_{1},\overline{x}_1) > \sigma(x_{2},\overline{x}_2)\\
    \frac{\delta \alpha_1}{\delta\alpha_1 + \alpha_2} u_1(x_{1}) + \frac{\alpha_2}{\delta \alpha_1 + \alpha_2} u_2(x_{2},)  & \text{if }  \sigma(x_{1},\overline{x}_1) < \sigma(x_{2},\overline{x}_2)\\
    \alpha_1 u_1(x_{1}) + \alpha_2 u_2(x_{2}) & \text{if } \sigma(x_{1},\overline{x}_1) = \sigma(x_{2},\overline{x}_2)
    \end{cases}
\end{equation}

Accordingly, when an attribute, such as $x_1$, is relatively more salient, its weight in the utility function ($\alpha_1$) is effectively increased to $\frac{\alpha_1}{\alpha_1 + \delta \alpha_2}$, while the weight on the less salient attribute, $\alpha_2$, is reduced to $\frac{\delta \alpha_2}{\alpha_1 + \delta \alpha_2}$.The denominator ensures that the salience weights sum to one. The parameter $\delta$ governs the strength of the salience effect. When $\delta = 1$, there is no salience distortion, and attributes are weighted according to their baseline importance. As $\delta$ decreases below one, relatively more salient attributes receive disproportionately greater weight, while less salient attributes are increasingly discounted.

\subsection{Causal Effects}

We use $Z_i$ to denote treatment. In survey experiments, the treatment typically determines the values of the attributes that respondents observe. Accordingly, the attribute vector $X_i(Z_i)$ can be viewed as a function of $Z_i$. The individual causal effect on utility, comparing $Z_i=z$ and $Z_i=z'$ in environment $E$, is defined as 
$$
ICE_i(E)=V_i(X_i(z);E)-V_i(X_i(z');E).
$$ where $=V_i(X_i(z);E)$ denotes the potential utility under treatment $z$.\footnote{For conjoint experiments, please see SI \ref{si:conjoint}. For factorial designs, the framework can be readily extended to accommodate multiple treatments.} 

The outcome variable need not directly measure this latent evaluation. Instead, it is generally a function of the observed response, given by $ Y_i = g_i(V_i(X_i;E))$. This formulation nests several commonly used outcomes: forced-choice decisions, where \(g_i(v)=\1\{v\ge 0\}\); ratings or thermometer scales, where \(g_i\) is the identity or a discretized monotone transformation; and approval or support indicators, where \(g_i\) follows a threshold rule. Because \(g_i\) is typically weakly increasing, the latent comparison in $V_i$ remains the central theoretical object. Our results extend to any monotone transformation of $V_i$. 

As discussed in the introduction, we examine whether causal effects identified in the survey environment generalize to the real-world environment while holding fixed the individuals, treatments, and outcome mappings. This constitutes a necessary precondition for any subsequent discussion of generalizability to a broader population. Following \citet{egami2023elements}, we focus on two key dimensions of generalizability: effect magnitude and effect sign.

\begin{definition}[Effect Magnitude Equivalence]
    $ICE_i(z,z';E^s)=ICE_i(z,z';E^r)$
\end{definition}

\begin{definition}[Effect Sign Congruence]
    $sign[ICE_i(z,z';E^s)]=sign[ICE_i(z,z';E^r)]$
\end{definition}

Generalizability of average effects follows naturally once we aggregate individual-level causal effects across the population. We focus on ICE because it provides a clear way to illustrate the mechanism at the individual level and, as discussed earlier, because our notion of weaker generalizability involves holding individuals, treatments, and outcome mappings fixed.

\section{Limited Attention and Survey Experiment Generalizability}
\label{sec:attention}

This section studies how attention mechanisms affect generalizability, with a particular focus on the equivalence of effect magnitudes. To isolate the role of attention, we first suppress salience distortions and examine how generalizability is affected when the survey environment alters only the active consideration set. Accordingly, as a benchmark, we assume $\delta=1$ throughout this section. The evaluation function then simplifies to $V_i(X_i,E)=\sum_{h \in C_i(E)}\alpha_{ik} u_{ik}(X_{ik})$, where the (effective) salience $\alpha_{ik}$ are invariant to attribute values.

Individuals have limited attention. In most decision-making contexts, it is unrealistic for individuals to consider all potentially relevant features. Instead, decisions are based only on attributes that enter the active consideration set. A defining feature of the survey environment is that respondents’ attention is strongly shaped by the information explicitly provided in the experimental design. For example, in a candidate-choice experiment where only age and gender are presented, respondents are likely to focus primarily, if not exclusively, on these attributes. Ensuring that respondents attend to the provided information, rather than skim past it, is therefore a central concern in survey experiments, motivating the widespread use of attention checks to detect inattentive respondents.

By contrast, in the real-world environment, individuals are not constrained by the limited set of attributes specified by the researcher. When making analogous decisions outside the survey context, respondents may attend to a broader and potentially different set of features. To formalize this distinction, we analyze the attention mechanism under the following limited-attention assumption.

\begin{assumption}[Limited Attention]\label{ass:attention}
$C_i(E^s) \subset C_i(E^r).$
\end{assumption} The consideration set in survey experiments is often a subset of the real-world consideration set for several reasons. In many applications, this provides a natural benchmark. First, survey experiments are designed to address specific research questions. As a result, researchers typically include a selected set of substantively relevant attributes, along with additional attributes serving as controls. These attributes are therefore best understood as a subset of the real-world consideration set $C_i(E^r)$. Second, researchers rarely possess complete knowledge of all attributes that respondents consider in real-world decision-making. Even if such knowledge were available, it would generally be infeasible to present the full set of relevant attributes within the constraints of a single survey instrument.

Ideally, respondents attend primarily to the attributes presented in the survey experiment, as these are the sources of randomized variation. However, in practice, respondents may also infer unlisted characteristics or rely on prior beliefs \citep{dafoe2018information}. A more general formulation therefore allows for $C_i(E^s) \neq C_i(E^r)$. In some special cases, it is even possible that $C_i(E^s) \supset C_i(E^r)$. Under such scenarios, our results continue to hold qualitatively, although the implications shift—from amplified effect magnitudes to potentially attenuated or unequal effects. Empirically, however, available evidence is more consistent with the limited-attention benchmark, although direct evidence remains limited. 

Without loss of generality, we assume that the real-world consideration set is $C_I^{r} = \{X_1, X_2, ..., X_m\}$, with associated baseline salience weights denoted by $\alpha^r=(\alpha_{i1},...,\alpha_{im})$. The experimental consideration set is given by $C_i^{s} = \{X_1, X_2, ..., X_k\}$, with corresponding baseline salience $\alpha^s=(\alpha'_{i1},...,\alpha'_{ik})$. We assume the cardinality of the real-world consideration set, $m=|C^r_i|$, exceeds that of the experimental consideration set, $k=|C_i^{s}|$. In other words, the experimental consideration set consists of only the first $k$ attributes from the real-world consideration set.

Now, we provide intuition for how limited attention affects the magnitude of causal effects. For ease of exposition, suppose the treatment in the survey experiment changes only the first attribute $X_{i1}(Z_i)$, from $x_{i1}$ to $x'_{i1}$. Then, the utility changes from $V_i(x_{i1},x_{i-1},E^s)=\alpha_{i1}u_{i1}(x_{i1})+ \sum_{j=2}^m \alpha_j u_{ij}(x_{ij})$ to $V_i(x'_{i1},,x_{i-1},E^s)=\alpha_{i1}u_{i1}(x'_{i1})+ \sum_{j=2}^m \alpha_j u_{ij}(x_{ij})$. Therefore, the ICE in the survey experiment is $ICE_I(E^S)=V_i(x_{i1},x_{i-1},E^s)-V_i(x'_{i1},x_{i-1},E^s)=\alpha_{i1}[u_{i1}(x_{i1})-u_{i1}(x'_{i1})]$. Similarly, in the real-world environment, the ICE is $\alpha'_{i1}[u_{i1}(x_{i1})-u_{i1}(x'_{i1})]$. Unless the baseline salience weights coincide, the two effects will generally differ. Under limited attention, the survey consideration set contains fewer attributes than the real-world consideration set ($C_i(E^s) \subset C_i(E^r)$). As a result, each attribute in the smaller consideration set $C_i(E^s)$ tends to receive relatively greater weight than in $C_i(E^r)$. Consequently, it is unlikely that $\alpha_{i1}=\alpha'_{i1}$. Instead, it is more plausible that $\alpha_{i1}>\alpha'_{i1}$, implying that the causal effect is amplified in the survey environment, $ICE(E^s)>ICE(E^r)$.

Formally establishing this amplification result requires additional assumptions. In particular, when the real-world consideration set includes additional attributes, these attributes may be correlated with existing ones—especially the focal attribute $X_{i1}$-which can alter baseline salience weights and potentially attenuate or even reverse the effect. Although such scenarios are possible, our focus is on isolating the limited-attention mechanism. Moreover, well-designed survey experiments rarely omit attributes that are highly correlated with the focal attribute, as doing so would undermine the internal validity of the design. To formalize this idea, we impose that the relative importance of attributes is preserved even when additional attributes are introduced.

\begin{assumption}\label{ass:stable}
For two consideration set $C^1_i$ and $C^2_i$, and any $j,k \in C^1_i \cap C^2_i$,
$$
\frac{\alpha_{ij}}{\alpha_{ik}}=\frac{\alpha'_{ij}}{\alpha'_{ik}}
$$
\end{assumption}
Specifically, if a respondent perceives attribute $X_1$ to be more important than attribute $X_2$ in one setting, then in a comparable setting with an expanded consideration set, the respondent continues to rank $X_1$ as more important than $X_2$.
This assumption is more likely to hold in well-designed experiments that explicitly incorporate attributes highly correlated with the focal dimensions of interest, thereby ensuring that any additional attributes in the real-world consideration set are not strongly correlated with those included in the experimental design.

The following results hold for general treatment.Let $D^i=\{j\in \{1,...,k\}|u_{ij}(X_{ij}(z)) \neq u_{ij}(X_{ij}(z'))\}$ denote the set of attribute indices for which the treatment affects individual $i$. 

\begin{proposition}[Effect Magnitude Non-Generalizability]\label{thm:ampbias}

Given treatment assignment $Z=z$ and $Z=z'$, suppose the limited attention assumption \ref{ass:attention} hold, and $D_i \neq \emptyset$.

(1)  The causal effects in the survey experiment and the real world differ almost surely: $$ICE(E^s) \neq ICE(E^s)$$ if $\alpha_{ij} \neq \alpha'_{ij}$ for some $j \in C^r_s$ and $u_{ij}(X_{ij}(z)) \neq u_{ij}(X_{ij}(z'))$.

(2)  Moreover, if the assumption \ref{ass:stable} also holds, then the causal effect in the survey experiment are amplified relative to those in the real-world environment by a factor \(\delta = \frac{1}{\alpha'_1 + \alpha'_2 + \ldots + \alpha'_k} > 1\).
\end{proposition}

\begin{proof}
    All proofs are in the SI.
\end{proof}

The first result follows directly. When baseline salience differs between the survey and real-world environments due to limited attention, the magnitude of causal effects estimated in the experimental setting will, in general, not be externally valid.

Under the additional stable salience assumption, we obtain a sharper result. If the inclusion of additional attributes in the real-world environment does not distort the relative importance of the attributes already present in the experiment, then the experimental causal effect is systematically amplified relative to the real-world effect. This condition is more likely to hold when the experimental design incorporates the most substantively important attributes—particularly those that are highly correlated with other relevant dimensions—so that omitted attributes do not substantially alter relative salience rankings.

The amplification factor \(\delta\) depends on the total salience of the attributes included in the survey experiment, as measured by their real-world salience weights
$\alpha$. When respondents place substantial weight on attributes that are omitted from the experimental design, the share of total salience allocated to the included attributes is correspondingly reduced in the real world. As a result, the experimental estimate which implicitly redistributes attention over a smaller set of attributes can substantially overstate the true causal effect.

\subsubsection*{Empirical Evidence}

The formal model and Proposition \ref{thm:ampbias} illustrate the mechanism through which limited attention distorts the magnitude of experimental causal effects. We now turn to empirical evidence. Our goal is not to “prove” the model—no formal model can be proven in that sense. The value of the model here lies in its ability to clarify and organize an empirically relevant phenomenon.

One implication of Proposition \ref{thm:ampbias} is that amplification bias should decline as more attributes are included in the consideration set. Of course, neither the true consideration set nor the real-world causal effect is directly observable. What researchers can manipulate, however, is the set of attributes presented in the survey experiment. When respondents attend to the experimental task, their attention is likely to be concentrated primarily on the attributes explicitly provided by the researcher. This is particularly true for conjoint experiments compared to other types of survey experiments. We therefore focus primarily on conjoint experiments. We argue that the attributes included in the experiment largely determine, or at least dominate, the active consideration set.

To examine this implication, we first study the 67 candidate-choice conjoint and vignette experiments compiled by \citet{schwarz2022have}. In these experiments, a gender attribute is randomly assigned, and the outcome measures are comparable or can be consistently recoded across studies. This common structure allows us to isolate and compare the average causal effect of the gender attribute across studies. In addition, we collected information on the number of non-gender attributes presented in each experiment. Notably, in some studies, the number of attributes provided is as small as two.

To investigate the hypothesized relationship between the number of attributes and the AMCE of gender across these experiments, we conduct a meta-analysis. Figure \ref{fig:meta} presents the meta-regression results based on a random-effects model, both for the full sample of countries and for the United States subsample, given that the majority of experiments were conducted in the United States. The vertical axis reports the absolute value of the estimated effect, recognizing that the causal effect of gender may be negative in some studies.\footnote{Consistent with our theoretical framework, a true positive effect is expected to diminish as the number of attributes increases, whereas a true negative effect should move toward zero in magnitude (i.e., its absolute value decreases).} Across both specifications, we find a statistically significant negative relationship between the number of attributes and the experimental causal effect of gender, consistent with our theoretical predictions. The results are robust to excluding experiments with a small number of negative estimates as well as potential outliers.

\begin{figure}[!ht]
    \centering
    \includegraphics[scale=0.7]{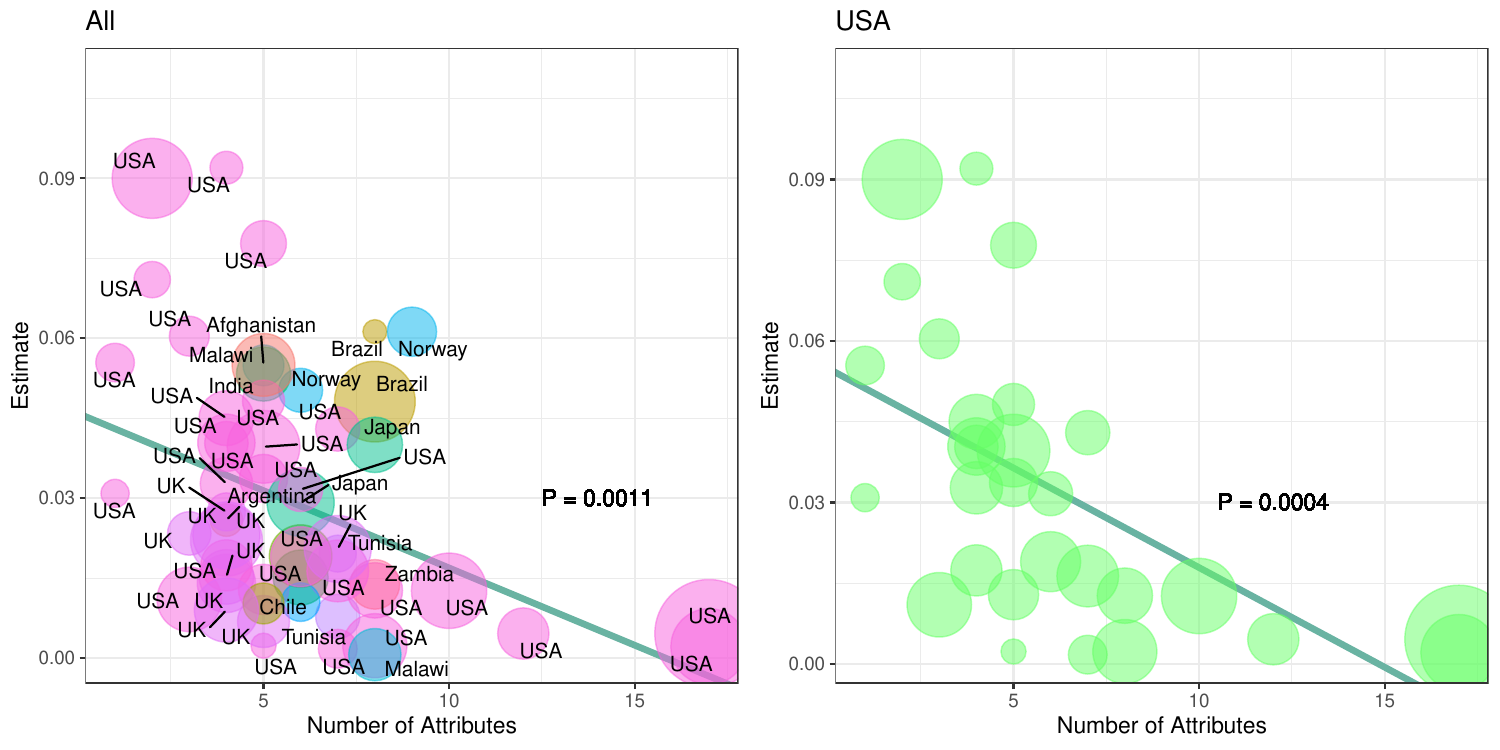}
    \caption{Test of the amplification hypothesis using meta-regression. Each point represents the estimated effect of gender from a conjoint experiment conducted in a given country, with point size proportional to the inverse of the estimate’s variance. The fitted meta-regression line indicates that the treatment effect decreases as the number of attributes increases. Full results are reported in Table \ref{tab:meta}, Columns 1–2.}
    \label{fig:meta}
\end{figure}

To address cross-study heterogeneity and more directly test our hypothesis regarding the influence of the number of attributes on experimental effects, we designed and conducted an original candidate-choice experiment with a controlled and consistent attribute set. The experiment was fielded to 1,200 respondents in the United States via Lucid in September 2023. Lucid employs a quota sampling strategy to align participant demographics with those of the U.S. Census. Prior research by \citet{coppock2019validating} shows that Lucid samples yield behavioral patterns comparable to those observed in nationally representative benchmark experiments.

Additional details on the experimental design are provided in SI \ref{si:design}. Respondents were randomly assigned to one of five groups. Within each group, participants completed six paired candidate-choice tasks. The groups differed in the number of attributes presented, while holding the specific attribute set fixed within each group.

\bigskip
\noindent \fbox{%
\parbox{\textwidth}{%
    - Group 1 was presented with only two attributes: gender and age.
    
    - Group 2 included the previous attributes plus education and tax policy, totaling four attributes.
    
    - Group 3 added race and income to the attributes in Group 2, resulting in six attributes.
    
    - Group 4 included military service and religious beliefs, bringing the total to eight attributes.
    
    - Group 5 encompassed ten attributes by adding children and marital status to those in Group 4.
}
}

\bigskip

Consistent with our meta-analytic focus, we examine the effect of gender across these groups. The results, presented in Figure \ref{fig:single}, reveal a statistically significant negative relationship between the number of attributes and the magnitude of the experimental gender effect, despite the limited number of observations (five groups). This pattern provides supporting evidence for our theory, which predicts that effect sizes diminish as the number of attributes increases.

An exception arises in Group 4 (with eight attributes), where the estimated gender effect is larger than expected. While this deviation may reflect sampling variability, an alternative explanation is that salience effects may be at play—a possibility we explore in subsequent sections.

\begin{figure}[!ht]
    \centering
    \includegraphics[scale=0.7]{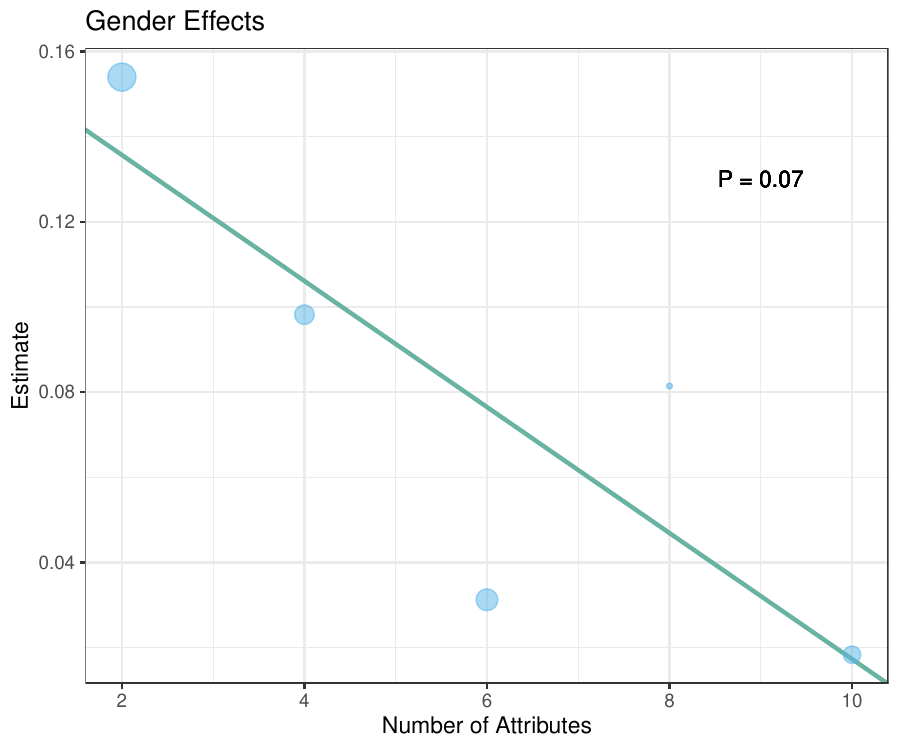}
    \caption{Test of the amplification hypothesis using a single conjoint experiment. Each point represents the estimated AMCE of the gender attribute under a given number of attributes, with point size proportional to the inverse of the estimate’s variance. The fitted line is obtained from a meta-regression. Full results are reported in Table \ref{tab:meta}, Column 3, and Table \ref{tab:full}.}
    \label{fig:single}
\end{figure}

\section{Salience and Survey Experiment Generalizability}\label{sec:salience}

In this section, we incorporate the role of salience. Building on the framework developed by \citet{bordalo2012salience}, we examine how salience, in conjunction with limited attention, further shapes the generalizability of survey experiments. In particular, we focus on the possibility of effect sign reversal across environments.

Not all attributes are equally important in decision-making. Evidence from survey experiments, including studies using eye-tracking techniques, shows that respondents process information selectively, focusing on attributes they perceive as important while ignoring less relevant ones as task complexity increases. For example, \citet{jenke2021using} demonstrate that respondents in conjoint experiments allocate attention unevenly across attributes. In salience theory, attention is differentially allocated to attributes that stand out relative to a reference point, as captured by the salience function $\sigma(x_{ij},\overline{X}_{ij})$. To operationalize this idea, we rank attributes by salience, where a smaller rank $r$ indicates greater salience. The resulting mechanism is summarized by equation \ref{eq:salience_weight}, which we reproduce here for convenience:
$$
\alpha_{ik}(X_i,E)
=
\frac{\bar{\alpha}_{ik}\,\delta^{\,r_{ik}(X_i,E)-1}}
{\sum_{j\in C_i(E)} \bar{\alpha}_{ij}\,\delta^{\,r_{ij}(X_i,E)-1}}
$$
Therefore, if an attribute lies close to its reference point, its salience is substantially discounted; conversely, attributes that deviate markedly from the reference receive relatively greater salience. The parameter $\delta$ governs the strength of this effect. To capture the possibility that perceived attribute salience varies across decision contexts, we introduce the following assumption:

\begin{assumption}[Salience Effect]\label{ass:salience}
    $\delta<1$.
\end{assumption}

The salience model emphasizes that the weight assigned to each attribute depends on the extent to which its realized value deviates from a reference point or prevailing expectation. Throughout this section, we maintain the assumption \ref{ass:stable}. Absent this condition, changes in the attribute set could alter relative salience rankings in ways that confound the mechanism of interest, making it difficult to disentangle salience-driven distortions from broader shifts in underlying preferences. 

\subsection{Effect Sign Reversal}

In the previous section, we showed that experimental estimates can be distorted by limited attention, leading to systematically amplified causal effects. Such amplification may, in some cases, be advantageous for researchers seeking to detect the direction of an effect, as it can reduce the sample size required to achieve a given level of statistical power.

A more conservative—and ultimately more fundamental—requirement for experimental evidence to meaningfully validate theoretical predictions is effect sign congruence: the experimental causal effect should have the same direction as the corresponding real-world effect \citep{slough2022sign}. However, once salience effects are introduced, this requirement may fail. In particular, salience-induced reweighting of attributes can generate effect sign reversal, thereby undermining the reliability of experimental findings as indicators of underlying causal relationships.

As is evident from equation \ref{eq:salience_weight}, salience depends on the realized attribute levels. Accordingly, the salience rank for attribute $k$ can be expressed as a function of the treatment, $r_{k}(Z_i)$. The following proposition provides a sufficient condition under which effect sign reversal may occur.

\begin{proposition}[Effect Sign Reversal]\label{prop:reverse2}
Suppose that assumptions \ref{ass:stable} and \ref{ass:salience} hold. The direction of the experimental effects may differ from their real-world counterparts if there exists an attribute $k$ such that $r_k(z_i) \neq r_k(z'_i)$.
\end{proposition}

The key condition in this proposition -- that there exists an attribute $k$ such that the salience ranking changes across treatment states -- captures a scenario in which a change in the treatment alters the relative salience of at least one attribute. As shown in the proof, such a reversal can arise regardless of the utility magnitudes associated with excluded attributes or the baseline salience of the treatment attribute itself.

To build intuition for the phenomenon of effect sign reversal, consider again a survey experiment in which the treatment affects only the first attribute. As illustrated in the left panel of Table \ref{tab:comp31}, under the control condition, the attribute values are $(x_1,x_2)$ with corresponding salience weights $(\alpha_1,\alpha_2)$. The treatment changes only $x_1$ to $\tilde{x}_1$, leading to updated salience weights $(\alpha'_1,\alpha'_2)$. In the real-world environment, suppose there is an additional attribute $X_3$ in the consideration set, as shown in the right panel of the table \ref{tab:comp31}. Under the same experiment, only the first attribute is affected. The individual causal effect in the survey environment is given by
$$ICE(E^s)=\sum_{k=1}^2\alpha_{ik}u_{ik}(x_k)-\alpha_{ik}u_{ik}(\tilde{x}_k)$$ where $\tilde{x}_2=x_2$ since the treatment only affects the first attribute. In contrast, in the real-world environment, the individual causal effect is
$$
ICE(E^r)=\sum_{k=1}^2\beta_{ik}u_{ik}(x_k)-\beta_{ik}u_{ik}(\tilde{x}_k)+[\beta_3u_{ik}(x_3)-\tilde{\beta}_3 u_{ik}(\tilde{x}_3)].
$$
We observe that salience affects the individual causal effect through two distinct channels. First, changes in salience weights imply that, even when the treatment does not alter the values of other attributes, those attributes still contribute to the treatment effect via reweighting. In particular, because the real-world environment includes an additional attribute, the salience weights assigned to the first two attributes $(\beta)$ differ from those in the survey environment $(\alpha)$. This reallocation of attention can lead to differences in effect magnitude and, potentially, in effect sign across environments. Second, in the real-world environment, $ICE(E^r)$ includes an additional term arising from the third attribute, $[\beta_3u_{ik}(x_3)-\tilde{\beta}_3 u_{ik}(\tilde{x}_3)]$. This additional component can be sufficiently large to overturn the direction of the treatment effect, thereby inducing effect sign reversal.

\begin{table}[htbp]
    \centering 
    \begin{tabular}{cc|cc|cc|cc}
    \hline
        \hline
    \multicolumn{4}{c}{Survey Experiment Environment}& \multicolumn{4}{c}{Real-world Environment}\\
        \hline
     \multicolumn{2}{c}{Control Status}  & \multicolumn{2}{c}{Treatment Status}  & \multicolumn{2}{c}{Control Status}  & \multicolumn{2}{c}{Treatment Status} \\
   
           Attribute &  Salience &  Attribute &  Salience &  Attribute &  Salience &  Attribute &  Salience \\
          \hline       
          \framebox{$x_1$} & $\alpha_1$ & \framebox{$\tilde{x}_1$} & $\alpha'_1$ &  \framebox{$x_1$} & $\beta_1$ & \framebox{$\tilde{x}_1$} & $\beta'_1$\\
          $x_2$ & $\alpha_2$ & $x_2$ & $\alpha'_2$ &  $x_2$ & $\beta_2$ & $x_2$ & $\beta'_2$ \\
     &  &  & &    $x_3$ & $\beta_3$ & $x_3$ & $\beta'_3$\\
         \hline
         \hline
    \end{tabular}
    \caption{Illustration of Salience Effect.}\label{tab:comp31}
\end{table}

It is immediate that if there exist parameter configurations under which a positive effect can be reversed to a negative one, then it is even more likely that configurations exist under which the effect is attenuated to zero. We formalize this observation in the following corollary:

\begin{corollary}\label{cor:reversal}
Suppose assumptions \ref{ass:stable} and \ref{ass:salience} hold. If there exists $j$ and $k$, such that $r_j(z_i) \neq r_k(z'_i)$, then the experimental treatment effect may become null in the real-world environment, and vice versa.
\end{corollary}

It is important to emphasize that our results are existence results; they do not imply that effect sign reversal or attenuation must occur in practice. For example, if a particular attribute—or its associated salience—dominates the evaluation process, then reversal is unlikely. The intuition is straightforward. Suppose attribute $X_1$ is substantially more important than all other attributes. In that case, it yield a large utility contribution $u(X_1)$. Consequently, in the overall evaluation $V$, the term $\alpha_1 u(X_1)$ dominates. Even if salience rankings shift, this dominant component is unlikely to be offset by changes in other attributes. As a result, the sign of $V$, and hence of the individual causal effect, will largely be determined by $X_1$, making sign reversal unlikely in such settings.

Another scenario that precludes effect sign reversal or attenuation arises when the rank-change condition is not satisfied. As emphasized in Proposition \ref{thm:ampbias}, variation in salience rankings across treatment states is essential for generating sign reversal. The following proposition formalizes that, in the absence of such rank changes, effect sign reversal cannot occur.

\begin{proposition}\label{prop:noreverse}
Assuming conditions \ref{ass:stable} and \ref{ass:salience} hold. If the relative salience rankings satisfy $r_k(z) = r_k(z')$ for all $k$, then the effect sign reversal cannot occur.
\end{proposition}

\subsection{Empirical Evidence}\label{sec:test_reversal}

The preceding propositions illustrate how salience can generate effect sign reversal or effect attenuation. We cannot directly test the theory, because neither the real-world consideration set nor the exact salience weights are observed or experimentally controlled. What we can do, instead, is derive testable implications. As before, our goal is not to “prove” the model.

Although we cannot manipulate the real-world consideration set, we do expect to substantially influence the consideration set in the survey environment. Building on Proposition \ref{prop:reverse2} and Corollary \ref{cor:reversal}, we therefore derive the following testable implications:

\begin{hypothesis}\label{hypo:reversal1}
 The sign of the effect in a conjoint experiment may reverse or attenuate to null as the number of attributes increases.
\end{hypothesis}

The proposition also highlights that changes in salience rankings are necessary for this phenomenon to arise. This raises the question: under what conditions can we expect salience rankings to remain unchanged, as assumed in the preceding proposition? The mechanism underlying salience implies that rankings are determined by the salience function $\sigma(x_k,\overline{X}_k)$. If the treatment-induced change in the attribute $X_1$ is not sufficiently large to meaningfully alter the reference point $\overline{X}_k$, then $\sigma(x_k,\overline{X}_k)$ will be close to $\sigma(\tilde{x}_k,\overline{X}_k)$. In such cases, the change in salience is too small to affect the relative ranking of attributes. Motivated by this intuition, we propose the following implication:

\begin{hypothesis}\label{hypo:reversal2}
Attribute effect sign reversal or attenuation is less likely when changes in attribute levels are marginal.
\end{hypothesis}

We draw on data from a conjoint experiment on hotel rooms conducted by \citet{bansak2021beyond}. The study identifies four core attributes of hotel rooms: “view from the room (ocean or mountain view), floor (top, club lounge, or gym and spa floor), bedroom furniture (1 king bed and 1 small couch or 1 queen bed and 1 large couch), and the type of in-room wireless internet (free standard or paid high-bandwidth wireless).” In addition to these core attributes, the authors include 18 supplementary attributes that are largely unrelated to the core set.

Respondents were asked to choose their preferred hotel room from 15 paired comparisons, where each profile included the four core attributes along with a randomly selected subset of additional attributes. As a result, respondents were randomly assigned to one of 11 experimental conditions, with profiles containing 4, 5, 6, 7, 8, 9, 10, 12, 14, 18, or up to 22 attributes.

Figure \ref{fig:reversal} presents a heatmap of the results.\footnote{Figure \ref{fig:sign_pvalue} in the SI reports statistical significance. Due to limited statistical power, relatively few estimates are statistically significant. As the sample size increases, confidence intervals would be expected to narrow, potentially revealing additional cases of statistically significant effect reversals. At the same time, null effects are also consistent with our theoretical predictions and align with implication \ref{hypo:reversal1}.} The horizontal axis indicates the number of attributes, while the vertical axis corresponds to attribute levels. Darker colors denote negative AMCEs, whereas lighter colors indicate positive AMCEs. The results provide support for our hypothesis on effect sign reversal (implication \ref{hypo:reversal1}). In particular, the estimated effects of some attributes—such as menu, bar, closet, and pillow—exhibit substantial instability across conditions. By contrast, other attributes, including View, Towels, and Internet, display considerable stability. This pattern is consistent with our theoretical framework: these attributes likely carry high baseline salience and utility, allowing them to dominate the evaluation process and remain robust to changes in the attribute set. This evidence also reinforces the classification of View and Internet as core attributes, as emphasized by \citet{bansak2021beyond}.

To test implication \ref{hypo:reversal2}, we adapted our candidate-choice experimental design. To better control for salience, rather than randomly assigning attribute levels across profiles, we constrained level differences to be minimal. For example, for the age attribute with five levels (40, 52, 60, 68, 75), only adjacent values were allowed to appear in a comparison. Thus, if one profile was assigned age 52, the other profile could only take values 40 or 60.

As in the previous experiment, this study was fielded to 1,200 U.S. respondents via Lucid in November 2023. Each respondent completed five paired choice tasks. Table \ref{tab:revsal} reports the AMCEs for gender under both the original design and the modified, reduced-salience design. When attribute levels were assigned randomly—without controlling for salience—the corresponding AMCEs (Column 2) exhibit a sign reversal as the number of attributes increases from six to eight. By contrast, Column 3 presents results from the reduced-salience design. A key observation is the absence of a clear effect sign reversal under this specification. We emphasize that this finding does not imply that researchers should universally restrict conjoint designs to adjacent attribute levels. Rather, this design choice is intended as a targeted test of the theoretical mechanism. 

\begin{table}[ht]
\centering
\begin{tabular}{ccc}
  \hline 
  \hline
  Num  & Salience Effect & Reduced Salience\\ 
  \hline
 2 &  0.15 & 0.45 \\ 
   4 &  0.10 & 0.10 \\ 
   6 & 0.03 & 0.15 \\ 
   8 &  \framebox{-0.08} & 0.10 \\ 
   10 &   \framebox{-0.02} & 0.27 \\ 
   \hline
   \hline
\end{tabular}
\caption{Gender Effect. Full results are in the Table \ref{tab:full} and \ref{tab:full1}}\label{tab:revsal}
\end{table}

\begin{figure}[!htp]
    \centering
    \includegraphics[scale=0.7]{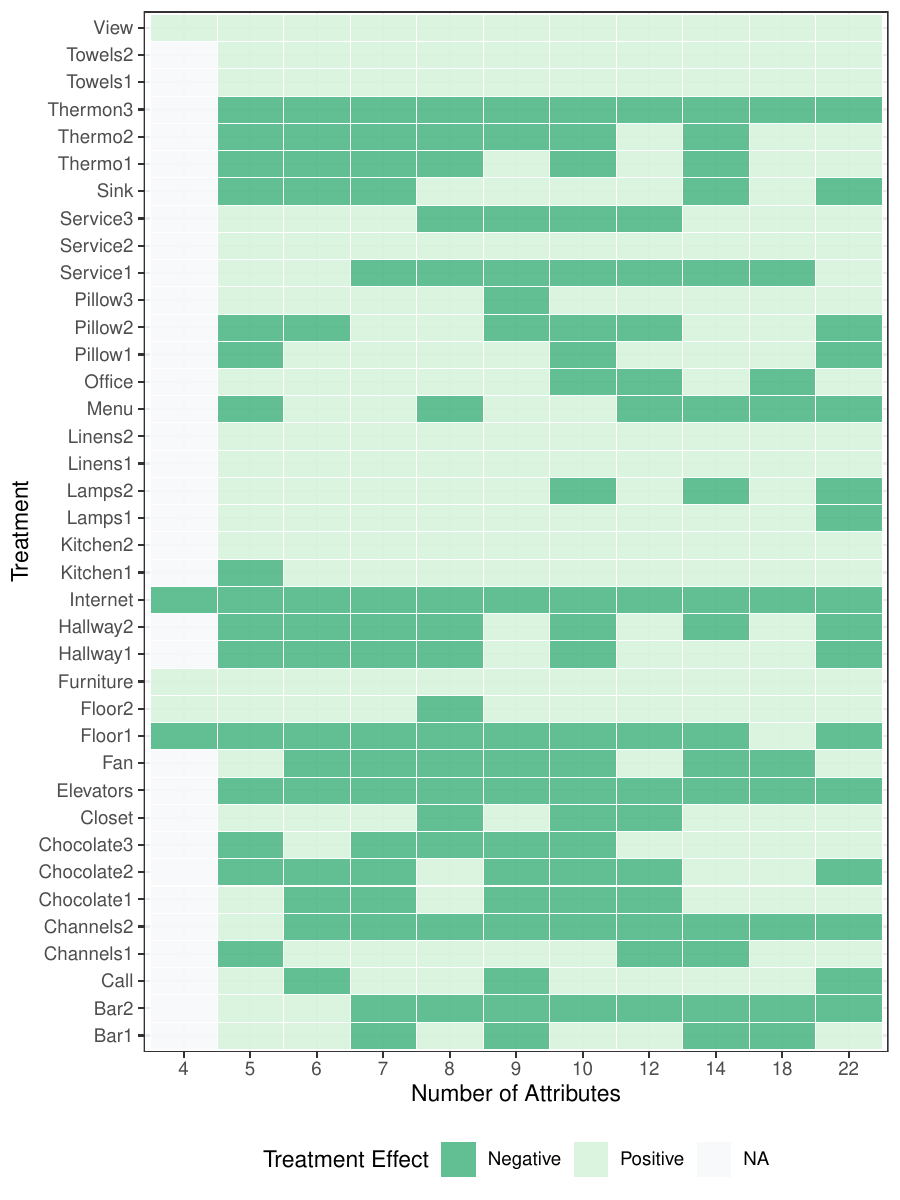}
    \caption{Testing the implication of effect sign reversal. The horizontal axis denotes the number of attributes, and each row represents the effect sign for a given attribute. Dark colors indicate negative effects, while light colors indicate positive effects. The 385 estimates are obtained from the same model specification as in \citet{bansak2021beyond}.}
    \label{fig:reversal}
\end{figure}


We emphasize again that our results are not claims about inevitability or frequency. We do not assert that such inconsistencies must occur, nor do we quantify how often they arise in practice. Rather, our goal is to illustrate a plausible mechanism. Other mechanisms—and potentially offsetting countervailing forces—may also be at work.

Indeed, based on our theory, relative to other types of survey experiments, well-designed conjoint experiments may be less susceptible to generalizability concerns. Because they incorporate multiple attributes within a single design, they substantially reduce the likelihood of omitting attributes that are important in real-world decision-making. The robustness of conjoint experiments has been widely discussed in the literature \citep{jenke2021using,hainmueller2015validating}.

\section{Discussion and Concluding Remarks}\label{sec:discuss}

An increasing number of studies in the social sciences employ experiments to identify causal effects (\citealt{druckman2006growth}). While such designs provide strong internal validity, external validity remains a longstanding concern. This paper develops a formal framework to explain why survey experiments may fail to satisfy this stronger notion of generalizability, even when treatment is randomized and even when the same individuals, treatment content, and outcome measures are held fixed. This limitation raises important concerns about the extent to which experimental findings can be extrapolated to real-world settings. For example, \citet{boas2019norms} show that although voters appear to sanction corruption information in survey experiments, they do not take action when presented with similar information about their own mayor in a field setting.

We provide both theoretical and empirical evidence demonstrating how two mechanisms—limited attention and salience effects—can undermine the generalizability of survey experiments. In particular, we show that experimental effects may be systematically amplified in magnitude and, in some cases, may even diverge in direction from their real-world counterparts.

A caveat to our findings is that we do not claim that non-generalizability is inevitable in survey experiments. Many studies—for example, \citet{jenke2021using}—as well as empirical evidence presented in this paper, suggest that certain attributes exhibit stability across contexts in conjoint experiments. Our theoretical framework also highlights that conjoint experiments possess particular advantages for achieving high levels of generalizability.



It is important to emphasize that salience effects are not inherently detrimental to experimental design. Whether individuals make decisions in real-world contexts or within experimental settings, salience is an integral component of multi-dimensional decision-making. Accordingly, the goal should not be to eliminate salience effects altogether. Rather, the objective is to ensure that the salience patterns induced in the experiment closely mirror those that arise in real-world environments. It is the artificially induced salience distortions—those that diverge from real-world conditions—that are most problematic.

Notably, in the absence of attention distortions—i.e., if the experimental environment perfectly replicates the real-world decision context—experimental ICEs would coincide with their real-world counterparts, precisely because salience patterns would align across environments. The challenge, of course, is that achieving such equivalence is rarely feasible in practice.

A more practical approach is to design attribute levels and profile combinations that closely approximate real-world scenarios. This calls for a phased research design that emphasizes realism, relevance, and precision. A natural starting point is a comprehensive review of the existing literature and available data on the decision-making context of interest. This preliminary step should be complemented by exploratory interviews or pilot surveys with individuals who engage in similar decisions in real-world settings. Such qualitative and descriptive evidence is essential for identifying the attributes and attribute levels that are most salient and substantively relevant to the decision-making process.

\newpage
\begin{spacing}{0.0}
	\bibliographystyle{apsr}
	\bibliography{cite1}
\end{spacing}


\clearpage
\setcounter{page}{1}

\appendix
\addcontentsline{toc}{section}{Appendix} 
\part{Supplementary Information} 
\parttoc 

\setcounter{figure}{0}
\setcounter{table}{0}
\renewcommand\thefigure{A.\arabic{figure}}
\renewcommand\thetable{A.\arabic{table}}

\clearpage

\section{Formal Definition of Causal Targets of Conjoint Experiment}\label{si:conjoint}

Consistent with standard practice, the causal effects are understood in terms of potential outcomes in a hypothetical scenario. We use Table \ref{tab:ite} to illustrate the concept. Suppose we aim to define the ICE for attribute $X_1$ when it changes from $x_1$ to $\tilde{x}_1$, while holding the other attributes constant. This involves considering two hypothetical scenarios.

\begin{table}[!h]
    \centering
    \begin{tabular}{ccc}
    \hline
    & \multicolumn{2}{c}{$World 1$}  \\
         &  $A_1$ &  $A_2$ \\
          \hline
      $X_1$ &  \framebox{$x_1$} & $x'_1$\\
      $X_2$ &   $x_2$ & $x'_2$\\
      $X_3$ &   $x_3$ & $x'_3$\\
         \hline
    \end{tabular}
    \quad 
\begin{tabular}{cc}
    \hline
     \multicolumn{2}{c}{$World 2$}  \\
           $\tilde{A}_1$ &  $A_2$ \\
          \hline
         \framebox{$\tilde{x}_1$} & $x'_1$ \\
         $x_2$ & $x'_2$\\
         $x_3$ & $x'_3$\\
         \hline
    \end{tabular}
     \caption{Illustration of Multi-dimensional Decisions. $A_j$ is the alternatives and $X_k$ is the attribute.}\label{tab:ite}
\end{table}

In the first scenario (World 1), decision makers (DMs) are presented with a choice between two candidates, $A_1$ and $A_2$. Each candidate is characterized by three attributes: $X_1, X_2$, and $X_3$. The realized values of these attributes for $A_1$ and $A_2$ are the vectors $(x_1, x_2, x_3)$ and $(x'_1, x'_2, x'_3)$, respectively. For example, if $X_1$ represents gender, with $x_1$ indicating female and $x'_1$ indicating male, the evaluations for the two candidates in World 1 are $V_i(A_1)$ and $V_i(A_2)$.

In the second scenario (World 2), the DM is presented with the same $A_2$, but for $A_1$, the attribute $X_1$ is realized as $\tilde{x}_1$ instead. In other words, all attribute realizations remain the same except for $X_1$. We thus use $\tilde{A}_1$ to differentiate this modified version of $A_1$ from its original. The potential evaluations in World 2 become: $V_i(\tilde{A}_1)$ and $V_i(A_2)$.

In practice, most conjoint experiments measure the forced choice outcome $Y$; that is, whether a candidate is chosen or not, denoted by $1$ or $0$. Theoretically, the DM chooses candidate $A_1$ over $A_2$ because the DM has a higher evaluation for the former. Because the utility difference drives the potential outcome we observe, without loss of generality, throughout the main text, we focus on $V$ rather than $Y$ \footnote{It is evident that the binary choice potential outcome can be defined as $Y_i^{j} = \mathbbm{1}[V^j_i \ge 0]$, where $j = 1, 2$. For outcomes that are preference scores, our results also hold because the score is also a function of the evaluation. In some survey experiments, the DM may only observe one alternative. In such cases, we can simply set $V_i(A_2) = 0$.}. Therefore, $Y$ is a function of the utility difference. The corresponding utility differences in World 1 are $V^1_i = V_i(A_1) - V_i(A_2)$, and in World 2, they are $V^2_i = V_i(\tilde{A}_1) - V_i(A_2)$ \footnote{Note that $V_i(A_2)$ in the two hypothetical worlds may differ due to the salience effect, which we will discuss in Section \ref{sec:salience}.}. Consequently, the individual causal effect for attribute $X_1$ when it changes from $x_1$ to $\tilde{x}_1$, given the other attributes remain constant, is defined as the difference-in-differences $V^1_i - V^2_i$. 

Our previous difference-in-differences $V^1_i - V^2_i$ corresponds to the individual component effect used in the conjoint experiment literature. See \citet{abramson2023detecting} for details.

To gain intuition regarding the phenomenon of Effect Sign Reversal, consider two hypothetical worlds, as depicted in the left panel of Table \ref{tab:comp3}. In Experiment 1, we compare $B_1$ and $B_2$, each with three attributes. In Experiment 2, the comparison involves two-attribute profiles $A_1$ and $A_2$. We are particularly interested in the treatment effect of attribute $X_1$ as it changes from $x_1$ to $\tilde{x}_1$. Consequently, $B_2$ (and $A_2$), serving as the controlled profiles, are fixed at $X^c=(x'_1,x'_2,x'_3)$ (and ($x'_1,x'_2$)) respectively.

We use $\beta_k$ ($\alpha_k$) to denote the salience of each attribute $x_k$ when the DM evaluates these alternatives. The corresponding sets of salience are illustrated in the right panel of Table \ref{tab:comp3}. Assume the existence of prior salience for each attribute, denoted by $\beta=(\beta^0_1,\beta^0_2,\beta^0_3)$ and $\alpha=(\alpha^0_1,\alpha^0_2)$. Consider how salience is formed and evolves as the DM observes the realized attributes during comparison. When an individual compares two profiles, differing levels of each attribute will ``distort" the original salience based on the rule previously mentioned. Specifically, in World 1, when comparing $B_1=(x_1,x_2,x_3)$ to $B_2=(x'_1,x'_2,x'_3)$, for each attribute $j$, we compute the salience function $\sigma_k(\cdot,\frac{x_k+x'_k}{2})$. The original salience $\beta^0$ is then discounted by $\delta^{r_k-1}$, where $r_k$ represents the relative salience rank introduced earlier. We assume that the updated salience attached to $B_1$ follows $\beta_1>\beta_2>\beta_3$ and the salience attached to $B_2$ follows $\beta'_1>\beta'_2>\beta'_3$.

In the hypothetical World 2, only the level of attribute 1 in the treatment group is altered, from $x_1$ to $\tilde{x}_1$. This modification affects the reference level for attribute $X_1$ and, consequently, the value of the salience function $\sigma_1$ for attribute 1. For instance, the updated salience attached to $\tilde{B}_1$ is now $\tilde{\beta}_2>\tilde{\beta}_1>\tilde{\beta}_3$; in other words, the salience of attribute 1 is now less than that of attribute 2. For simplicity, we assume that the salience values for the controlled profile remain the same as in the first comparison.

\begin{table}[]
    \centering 
    \begin{tabular}{cc|cc}
    \hline
    \multicolumn{4}{c}{Experiment 1}\\
     \multicolumn{2}{c}{World 1}  & \multicolumn{2}{c}{World 2} \\
   
           $B_1$ &  $B_2$ &  $\tilde{B}_1$ &  $B_2$ \\
          \hline
          
          \framebox{$x_1$} & $x'_1$ & \framebox{$\tilde{x}_1$} & $x'_1$\\
          $x_2$ & $x'_2$ & $x^t_2$ & $x'_2$ \\
         $x_3$ & $x'_3$ & $x^t_3$ & $x'_3$\\
         \hline
    \end{tabular}
    \quad 
    \begin{tabular}{cc|cc}
    \hline
     \multicolumn{4}{c}{Experiment 1}\\
     \multicolumn{2}{c}{ World 1}  & \multicolumn{2}{c}{ World 2} \\
    
           $\beta_1$ &  $\beta_2$ &  $\tilde{\beta_1}$ &  $\beta_2$ \\
          \hline
          
          $\beta_1$ & $\beta'_1$ & $\tilde{\beta}_1$ & $\beta'_1$\\
          $\beta_2$ & $\beta'_2$ & $\tilde{\beta}_2$ & $\beta'_2$ \\
         $\beta_3$ & $\beta'_3$ & $\tilde{\beta}_3$ & $\beta'_3$\\
                 \hline
    \end{tabular}

     \begin{tabular}{cc|cc}
    \hline
    \multicolumn{4}{c}{Experiment 2}\\
     \multicolumn{2}{c}{ World 1}  & \multicolumn{2}{c}{ World 2} \\
   
           $A_1$ &  $A_2$ &  $\tilde{A}_1$ &  $A_2$ \\
          \hline
          
          \framebox{$x_1$} & $x'_1$ & \framebox{$\tilde{x}_1$} & $x'_1$\\
          $x_2$ & $x'_2$ & $x^t_2$ & $x'_2$ \\
         \hline
    \end{tabular}
    \quad 
    \begin{tabular}{cc|cc}
    \hline
     \multicolumn{4}{c}{Experiment 2}\\
     \multicolumn{2}{c}{ World 1}  & \multicolumn{2}{c}{ World 2} \\
    
           $\alpha_1$ &  $\alpha_2$ &  $\tilde{\alpha_1}$ &  $\alpha_2$ \\
          \hline
          
          $\alpha_1$ & $\alpha'_1$ & $\tilde{\alpha}_1$ & $\alpha'_1$\\
          $\alpha_2$ & $\alpha'_2$ & $\tilde{\alpha}_2$ & $\alpha'_2$ \\
                 \hline
    \end{tabular}
    \caption{Illustration of Salience Effect. $A_j$ and $B_j$ are alternatives, $x_k$ denotes the attribute, and $\alpha_k$ and $\beta_k$ denotes thesalience.}\label{tab:comp3}
\end{table}

For Experiment 1, the utility differences in hypothetical World 1 are given by $\sum_{k} \beta_k u_k(x_k) - \sum_k \beta'_k u_k(x'_k)$, and in hypothetical World 2, they are $\sum_{j} \tilde{\beta}_k u_k(\tilde{x}_k) - \sum_k \beta'_k u_k(x'_k)$, where $\tilde{x}_2=x_2$ and $\tilde{x}_3=x_3$. Consequently, the individual treatment effect of attribute $X_1$ can be expressed as:
$$
\begin{aligned}
    ICE_3 &=\sum_{k=1}^3 [\beta_k u_k(x_k)-\tilde{\beta}_k u_k(\tilde{x}_k)]\\
    &=\sum_{k=1}^2 [\beta_k u_k(x_k)-\tilde{\beta}_k u_k(\tilde{x}_k)] + [\beta_3 u_3(x_3)-\tilde{\beta}_3 u_3(\tilde{x}_3)]
\end{aligned}
$$ where the subscript $3$ indicates that a total of three attributes are considered in Experiment 1. Similarly, the individual treatment effect for Experiment 2 is
$$
ICE_2 = \sum_{k=1}^2 [\alpha_k u_k(x_k)-\tilde{\alpha}_k u_k(\tilde{x}_k)]
$$

It is evident that $ICE_3$ includes an additional term $[\beta_3 u_3(x_3) - \tilde{\beta}_3 u_3(\tilde{x}_3)]$ compared to $ICE_2$. Furthermore, the salience values denoted by $\beta$ in $ICE_3$ differ from those denoted by $\alpha$ in $ICE_2$. These two factors could lead to $ICE_3$ and $ICE_2$ having different signs.

\begin{figure}[!htp]
    \centering
    \includegraphics[scale=0.8]{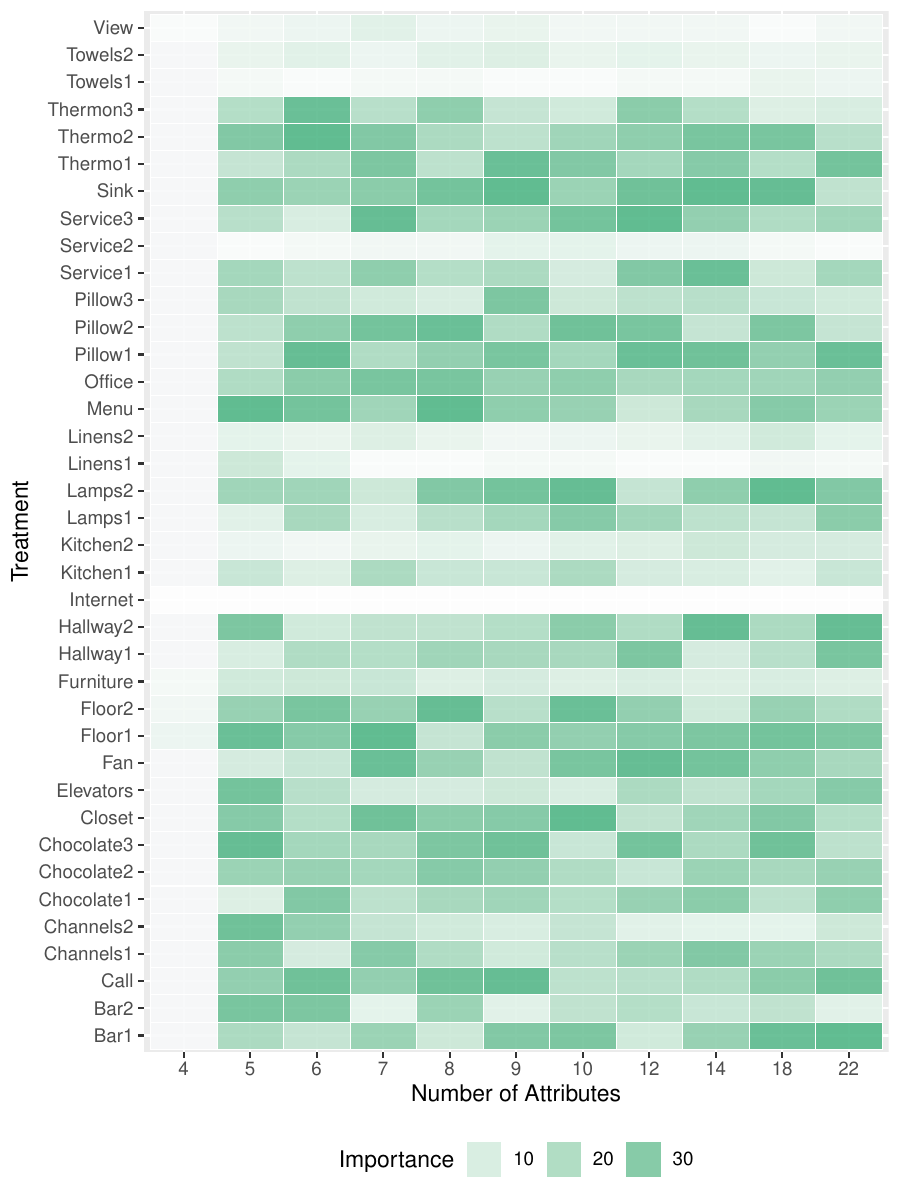}
    \caption{Testing Hypothesis of Importance Reversal. The horizontal line denotes the number of attributes and each row illustrates the relative importance of each attribute. The darker color represents less importance. The total 385 estimates are from the same model in \citet{bansak2021beyond}.}
    \label{fig:rsqu}
\end{figure}

 \label{si:formulation}

\section{Design of Candidate-Choice Experiment}\label{si:design}

This study is pre-registered with the center for Open Science Framework (OSF). An anonymized version along with the questionnaire can be found at the \href{https://osf.io/4v59z/?view_only=df0f02c8cf3e4e6bbd072460c421bfdd}{OSF}. The survey experiment received full Institutional Review Board (IRB) approval in Study ID [details removed for anonymous peer review] from the authors' institutions.

\subsection{Principles for Human Subjects Research}
We recruit subjects using the online survey platform, Lucid Theorem, which  manages relationships with suppliers who handle incentives to participants directly. Researchers pay Lucid a cost per (currently at \$1.50 for a 10--15 minute survey) for a completed interview (CPI) and Lucid pays suppliers who then provide a portion of those
earnings to participants in the form of cash, gift cards, or loyalty reward points. Lucid Theorem uses a proportional sampling method to provide nationally representative samples, balancing participants based on age, gender, ethnicity, and region. 

Each solicited respondent received a link from the company that redirected them to the actual survey and the consent information, hosted on a server maintained by Qualtrics. Upon completing the survey, respondents are immediately redirected back to the company’s website to claim their reward. All participants received compensation after completing the questionnaires.

In the beginning of the survey, we provided a consent form that laid out clear and comprehensive information about the project, including how data will be collected, used, and stored, the costs and benefits of participation, and contact information for the lead researcher and university IRB. Respondents were also informed that no identifying information would be collected, and that they were able to opt out of the research at any time. Each participant was required to read the consent form before proceeding with the survey (at which point consent was assumed to be granted). 

The project did not involve any deception, nor did we intervene in any political processes. The data were obtained in compliance with all relevant regulations and ethical guidelines. The raw data is handled exclusively by the authors, following strict protocols to ensure confidentiality and data security.

\subsection{Pre-registered Design}

The online anonymized version along with the questionnaire can be found at the \href{https://osf.io/4v59z/?view_only=df0f02c8cf3e4e6bbd072460c421bfdd}{OSF}. 

This study explores how the causal effects derived from conjoint experiments can be inconsistent with the real-world effect due to attention and salience biases. Using a series of candidate choice conjoint experiments, we will randomly vary the numbers of attributes (attention bias) and their levels (salience bias) presented to the respondents and examine the differences in the average marginal component effects (AMCEs) of two key features--gender and age--common across all experimental groups. We expect that the AMCEs will be smaller the more number of attributes are included and larger when the levels of the attributes are more salient.

We conduct two candidate-choice experiment to test hypotheses outlined in the main text. The sample sizes for both experiments were 1,200, which is determined by our research budget.

\textbf{Study Information}

Hypotheses

H1: The AMCEs of a common attribute will be smaller the more number of attributes are included in the conjoint experiment

H2: The AMCEs of a common attribute will be larger when the levels of the attributes are more salient in the conjoint experiment.

\textbf{Part 1:}

In the firs experiment, respondents see a conjoint table with a hypothetical candidate that is described by K (=2,4,6,8,10) attributes as shown in the example table below:

\begin{figure}
    \centering
    \includegraphics{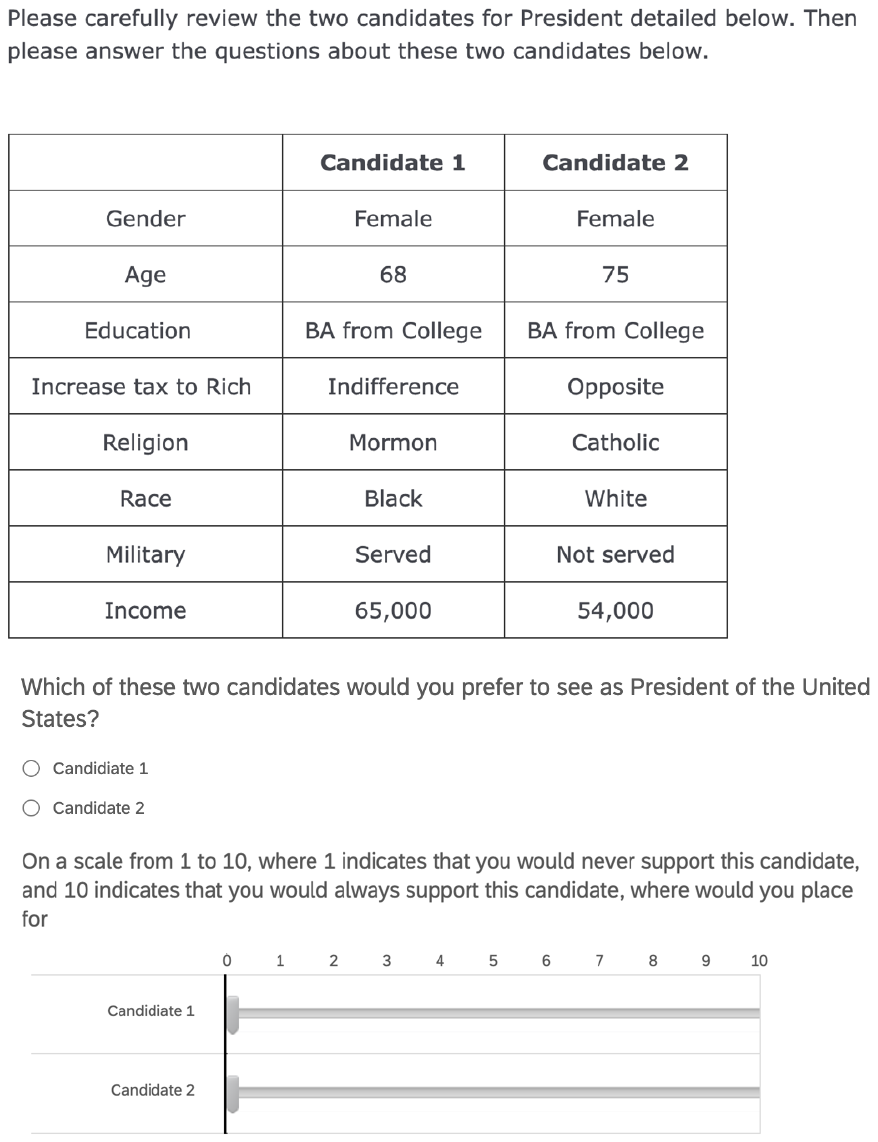}
    \caption{Caption}
    \label{fig:exp_design1}
\end{figure}

We randomly assign participants into one of five groups. For each attribute, levels are also randomly assigned. Each participant completed six rounds of the candidate choice experiment. 

\begin{itemize}
    \item Group 1: each hypothetical candidate has K=2 attributes including Gender and Age; there are total of 6 rounds for each participant
    \item Group 2: each hypothetical candidate has K=4 attributes including Gender, Age, Education, and Tax; there are total of 6 rounds for each participant
    \item Group 3: each hypothetical candidate has K=6 attributes including Gender, Age, Education, Tax, Race, and Income; there are total of 6 rounds for each participant
    \item Group 4: each hypothetical candidate has K=8 attributes including Gender, Age, Education, Tax, Race, Income, Religion, and Military service; there are total of 6 rounds for each participant
    \item Group 5: each hypothetical candidate has K=10 attributes including Gender, Age, Education, Tax, Race, Income, Religion, Military service, Gay Marriage, and Children; there are total of 6 rounds for each participant
\end{itemize}

Attribute level:

\begin{figure}
    \centering
    \includegraphics[width=0.9\textwidth]{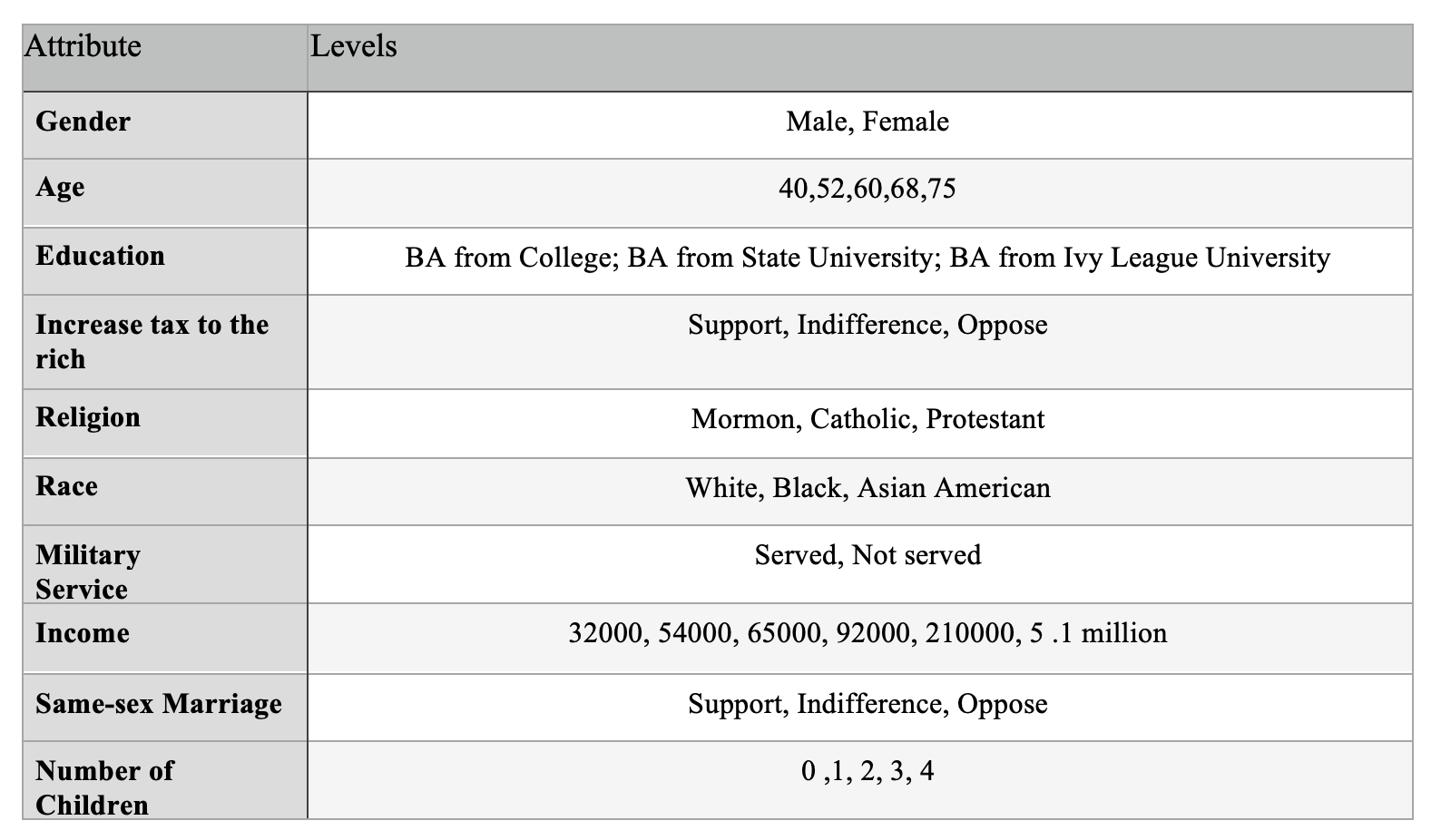}
    \caption{Attribute Table}
    \label{fig:table_attribute}
\end{figure}

\textbf{Part 2:}

In the second experiment, respondents are required to complete 5 rounds of candidate choice experiment similar to the first experiment. However, we control attribute salience in each pair. In particular, we let some attributes be randomly chosen so that their levels are close to each other. For example, for attribute age = [40,52,60,68,75], only levels (40,52) , (52,60), (60,68), and (68,75) will be randomly chosen. Those levels are age, income, same sex marriage, tax, education, and children.

\begin{itemize}
    \item If K=2, attributes are Gender and Age
 \item If K=4, attributes are Gender, Age, Tax, and Education.
 \item If K=6, attributes are Gender, Age, Tax, Education, Race, and Income.
 \item If K=8, attributes are Gender, Age, Tax, Education, Race, Income, Military, and Religion.
 \item If K=10, attributes are Gender, Age, Tax, Education, Race, Income, Military, Religion, Same Sex Marriage, and Children.
\end{itemize}

\textbf{Analysis:}

(1)	We will calculate the average marginal component effect (binary and 10-point rate), and  $R^2$ for each attribute in each K group. We also present both results with and without covariate adjustment. 

(2)	Check the relationship between AMCE and the number of attributes K by linear regression and meta-regression.

(3)	Check the relationship between the treatment effect direction (positive, negative) and the number of attributes K by drawing a heatmap.


(4)	Check the relationship between the treatment effect direction and the number of attributes K conditional on whether levels are salient or not.

\section{Proof of Proposition \ref{thm:ampbias}}







\begin{proof}

Recall, without loss of generality, we assume that the real-world consideration set is $C_I^{r} = \{X_1, X_2, ..., X_m\}$, the associated baseline salience is denoted as $(\alpha_{i1},...,\alpha_{im})$. The experimental consideration set is $C_i^{s} = \{X_1, X_2, ..., X_k\}$ and baseline salience is $(\alpha'_{i1},...,\alpha'_{ik})$. We assume the cardinality of $X^{r}$, denoted by $m$, is greater than that of $X^{e}$, denoted by $k$ ($|C_i^{r}| = m > |C_i^{s}| = k$).

Recall, $D^i=\{j\in[1,k]|u_{ij}(X_{ij}(z)) \neq u_{ij}(X_{ij}(z'))\} \neq \emptyset$ is the set of attribute index that treatment has effects on individual $i$.

Under treatment $Z$, the outcome is $V_i(z,E)=\sum_{j \in C_i\cap D_i}\alpha_{ij}u_{ij}(X_{ij}(Z))$. Therefore, the ICE in the survey experiment is 
$$
ICE(z,z';E^s)=\sum_{j \in C^s_i \cap D_i}\alpha_{ij}(u_{ij}(X_{ij}(z))-u_{ij}(X_{ij}(z'))).
$$
Note that treatment cannot affect the attributes that are not in the $C^s_i$. By assumption \ref{ass:attention}, $C^r_i \cap C^s_i=\cap C^s_i$. Therefore, the ICE in the real world is 
$$
ICE(z,z';E^r)=\sum_{j \in C^s_i \cap D_i}\alpha'_{ij}(u_{ij}(X_{ij}(z))-u_{ij}(X_{ij}(z'))).
$$

Fix vectors $(\alpha_{ij})_{j\in C_i\cap D_i}$ and $(\alpha'_{ij})_{\in C_i\cap D_i}$, where some $\alpha_{ij} \neq \alpha'_{ij}$ for some $j\in C_i\cap D_i$. Because $u_{ij}(X_{ij}(z))-u_{ij}(X_{ij}(z'))$ are not constrained, we treat them as a random draw from an closed interval, and thus the first result holds.

Under assumption \ref{ass:salience}, fixing a $j \in \{ C_i\cap D_i\}$, denote $\delta_{jk}=\frac{\alpha_{ik}}{\alpha_{ij}}=\frac{\alpha'_{ik}}{\alpha'_{ij}}$ for any $k \in \{C_i\cap D_i\}$. Therefore,
$$
ICE(z,z';E^s)= \alpha_{ij}\sum_{k \in C^s_i \cap D_i}\delta_{jk} \Delta u_{ik}
$$ and
$$
ICE(z,z';E^r)= \alpha'_{ij}\sum_{k \in C^s_i \cap D_i}\delta_{jk} \Delta u_{ik}
$$ where we use $\Delta u_{ik}$ to denote $u_{ij}(X_{ij}(z))-u_{ij}(X_{ij}(z'))$.

By assumption \ref{ass:salience} again, $\alpha'_{ij}=\alpha_{ij} (1-\sum_{j=k+1}^m\alpha'_{ij} )$. Therefore,
$$
\frac{ICE(z,z';E^s)}{ICE(z,z';E^r)}=\frac{1}{\sum_{j=1}^k \alpha'_{ij}}
$$

Now, we have shown results based on ICE. ACE holds because it is a weighted average of ICE.
\end{proof}

\section{Proof of Proposition \ref{prop:reverse2}}
\begin{proof}
We show existence for the case in which $|C_i^s|=2$ and $|C_i^r|=3$. The general case follows the same logic. Suppose there are two attributes in the consideration set of the survey experiment. Under treatment status $Z_i=z$, we denote the vector by $x_1=(u_{i1}(X_{i1}(z)),u_{i2}(X_{i2}(z)),0)$, and the corresponding salience is $a=(\alpha_{i1},\alpha_{i2},0)$. Under treatment $Z_i=z'$, we write $x'_1=(u_{i1}(X_{i1}(z')),u_{i2}(X_{i2}(z')),0)$, with corresponding salience $a'=(\alpha'_{i1},\alpha'_{i2},0)$. Under the assumption, $r_1(z_i) \neq r_1(z'_i)$, and therefore $a \neq a'$. Without loss of generality, we assume that the ICE is positive:
\begin{equation}\label{sieq1}
    ICE(E^s)=a\cdot x_1-a'\cdot x'_i >0
\end{equation}

    Similarly, suppose there are three attributes in the consideration set of the real-world experiment. We define $x_2=(u_{i1}(X_{i1}(z)),u_{i2}(X_{i2}(z)),u_{i3}(X_{i3}(z)))$, with corresponding salience $b=(\beta_{i1},\beta_{i2},\beta_{i3})$. We define $x'_2$ and $b'$ analogously. Therefore, we want to show that
\begin{equation}\label{sieq2}
ICE(E^r)=b \cdot x_2 - b' \cdot x'_2<0
\end{equation}

Note that $x_2=x_1+u_3$ and $x'_2=x'_1+u_3$, where $u_3=(0,0,u_{i3}(x_{i3}))$.

By stable salience, we define $k=\frac{\beta_{i1}}{\alpha_{i1}}$ and $k'=\frac{\beta'_{i1}}{\alpha'_{i1}}$.

Suppose that $k>k'$. Then we can write
$$
b=ka+b_3 \text{ and } b'=k'a'+b'_3 
$$
where $b_3=(0,0,\beta_{i3})$ and $b'_3=(0,0,\beta'_{i3})$. Substituting these into equation $\eqref{sieq2}$, we obtain
$$
\begin{aligned}
    (ka+b_3)\cdot (x1+u_3) - (k'a'+b'_3)\cdot (x'_1+u_3) & <0 \\
    ka \cdot x_1 - k'a' \cdot x'_1 +(\beta_{i3}-\beta'_{i3}) \cdot u_{i3}(x_{i3}) & <0
\end{aligned}
$$
Note that, by stable salience, we have $\beta_{i3}=1-k$ and $\beta'_{i3}=1-k'$. Substituting these into the inequality above, we obtain
$$
k(a\cdot x_1-u_{i3}(x_{i3}))<k'(a'\cdot x'_1-u_{i3}(x_{i3}))
$$
Also, combining this with equation \eqref{sieq1}, we conclude that any $u_{i3}(x_{i3})$ satisfying
$$
a' \cdot x'_1 < a \cdot x_1<u_{i3}(x_{i3})
$$ 
$$
\frac{a'_1 \cdot x'_1 - u_{i3}(x_{i3})}{a_1 \cdot x_1-u_{i3}(x_{i3}) } >1
$$
can produce effect sign reversal. For $k<k'$, we simply reverse the inequality sign.

As demonstrated in the proof, we have more “parameters” than constraints, and thus there is considerable flexibility in achieving the desired outcome. Identifying even a single scenario that aligns with our hypothesis is sufficient.
\end{proof}

\section{Proof of Proposition \ref{prop:noreverse}}

\begin{proof}
    Following the proof of Proposition \ref{thm:ampbias}, we have
    $$
\frac{ICE(z,z';E^s)}{ICE(z,z';E^r)}=\frac{1}{\sum_{j=1}^k \alpha'_{ij}}
$$

Because $\sum_{j=1}^k \alpha'_{ij}>0$, we have $\operatorname{sign}(ICE(z,z';E^s)) = \operatorname{sign}(ICE(z,z';E^r))$.
\end{proof}

 \label{si:proof}

\section{Tables}

\begin{table}[ht] \centering 
\begin{threeparttable}
  \caption{Gender Effects Meta-regression} 
  \label{tab:meta} 
\begin{tabular}{@{\extracolsep{-1pt}}lccc} 
\\[-1.8ex]\hline 
\hline \\[-1.8ex] 
\\[-1.8ex] & All & USA & Own\\ 
\\[-1.8ex] & (1) & (2) & (3) \\ 
\hline \\[-1.8ex] 
 Number & $-$0.0029$^{***}$ & $-$0.0037$^{***}$ & $-$0.0148$^{*}$\\ 
  & (0.0009) & (0.0037) & (0.0054)\\ 
  Intercept & 0.0459$^{***}$ & 0.0549$^{***}$ & 0.1653$^{**}$\\ 
   & (0.0059) & (0.0071) & (0.0352)\\ 
  \hline
  $R^2$ & 0.21 &0.43 & 0\\
  $\tau^2$ & 0.0003  & 0.0003 & 0\\
\hline 
\hline \\[-1.8ex] 
\end{tabular} 
{{Notes:} $R^2$ denotes the amount of heterogeneity accounted for and $\tau^2$ is the estimated amount of residual heterogeneity. 
 $^{*}$p$<$0.1; $^{**}$p$<$0.05; $^{***}$p$<$0.01.}
\end{threeparttable}
\end{table}

\begin{table}[]
\resizebox{15cm}{!}{
\begin{tabular}{@{\extracolsep{5pt}}lccccc} 
\\[-1.8ex]\hline 
\hline \\[-1.8ex] 
 & \multicolumn{5}{c}{Dependent Variable: Preference Score} \\ 
\cline{2-6} 
\\[-1.8ex] & (1) & (2) & (3) & (4) & (5)\\ 
\hline \\[-1.8ex] 
 GenderMale & 0.154$^{**}$ & 0.098 & 0.031 & $-$0.081 & $-$0.018 \\ 
  & (0.063) & (0.068) & (0.067) & (0.071) & (0.069) \\ 
  & & & & & \\ 
 Age52 & 0.119 & $-$0.041 & 0.032 & 0.004 & $-$0.131 \\ 
  & (0.099) & (0.107) & (0.107) & (0.112) & (0.108) \\ 
  & & & & & \\ 
 Age60 & $-$0.162 & $-$0.119 & $-$0.194$^{*}$ & $-$0.048 & $-$0.110 \\ 
  & (0.101) & (0.107) & (0.106) & (0.111) & (0.109) \\ 
  & & & & & \\ 
 Age68 & $-$0.425$^{***}$ & $-$0.331$^{***}$ & $-$0.273$^{**}$ & $-$0.145 & $-$0.253$^{**}$ \\ 
  & (0.101) & (0.107) & (0.107) & (0.112) & (0.108) \\ 
  & & & & & \\ 
 Age75 & $-$1.340$^{***}$ & $-$0.830$^{***}$ & $-$0.566$^{***}$ & $-$0.551$^{***}$ & $-$0.394$^{***}$ \\ 
  & (0.101) & (0.107) & (0.107) & (0.112) & (0.110) \\ 
  & & & & & \\ 
 EducationBA from Ivy League University &  & 0.048 & $-$0.027 & $-$0.042 & 0.045 \\ 
  &  & (0.083) & (0.082) & (0.087) & (0.084) \\ 
  & & & & & \\ 
 EducationBA from State University &  & 0.076 & 0.001 & $-$0.036 & 0.026 \\ 
  &  & (0.083) & (0.081) & (0.087) & (0.085) \\ 
  & & & & & \\ 
 TaxOppose &  & $-$0.327$^{***}$ & $-$0.374$^{***}$ & $-$0.209$^{**}$ & $-$0.094 \\ 
  &  & (0.083) & (0.082) & (0.087) & (0.085) \\ 
  & & & & & \\ 
 TaxSupport &  & 0.828$^{***}$ & 0.401$^{***}$ & 0.213$^{**}$ & 0.327$^{***}$ \\ 
  &  & (0.083) & (0.081) & (0.087) & (0.084) \\ 
  & & & & & \\ 
 RaceBlack &  &  & 0.137$^{*}$ & 0.116 & 0.037 \\ 
  &  &  & (0.082) & (0.087) & (0.084) \\ 
  & & & & & \\ 
RaceWhite &  &  & 0.066 & 0.163$^{*}$ & 0.078 \\ 
  &  &  & (0.082) & (0.087) & (0.084) \\ 
  & & & & & \\ 
 Income32,000 &  &  & $-$0.080 & 0.009 & $-$0.159 \\ 
  &  &  & (0.117) & (0.123) & (0.119) \\ 
  & & & & & \\ 
 Income5.1 million &  &  & $-$0.347$^{***}$ & $-$0.085 & $-$0.179 \\ 
  &  &  & (0.116) & (0.124) & (0.120) \\ 
  & & & & & \\ 
 Income54,000 &  &  & $-$0.030 & 0.089 & 0.075 \\ 
  &  &  & (0.116) & (0.123) & (0.119) \\ 
  & & & & & \\ 
 Income65,000 &  &  & $-$0.068 & 0.252$^{**}$ & 0.007 \\ 
  &  &  & (0.116) & (0.123) & (0.119) \\ 
  & & & & & \\ 
 Income92,000 &  &  & $-$0.051 & 0.095 & $-$0.030 \\ 
  &  &  & (0.117) & (0.123) & (0.120) \\ 
  & & & & & \\ 
 MilitaryServed &  &  &  & 0.209$^{***}$ & 0.134$^{*}$ \\ 
  &  &  &  & (0.071) & (0.069) \\ 
  & & & & & \\ 
  ReligionMormon &  &  &  & $-$0.429$^{***}$ & $-$0.173$^{**}$ \\ 
  &  &  &  & (0.087) & (0.084) \\ 
  & & & & & \\ 
 ReligionProtestant &  &  &  & $-$0.024 & 0.040 \\ 
  &  &  &  & (0.087) & (0.084) \\ 
  & & & & & \\ 
 Children1 &  &  &  &  & 0.073 \\ 
  &  &  &  &  & (0.107) \\ 
  & & & & & \\ 
 Children2 &  &  &  &  & 0.159 \\ 
  &  &  &  &  & (0.110) \\ 
  & & & & & \\ 
Children3 &  &  &  &  & 0.138 \\ 
  &  &  &  &  & (0.107) \\ 
  & & & & & \\ 
 Children4 &  &  &  &  & 0.158 \\ 
  &  &  &  &  & (0.109) \\ 
  & & & & & \\ 
 MarriageOppose &  &  &  &  & $-$0.481$^{***}$ \\ 
  &  &  &  &  & (0.084) \\ 
  & & & & & \\ 
 MarriageSupport &  &  &  &  & 0.010 \\ 
  &  &  &  &  & (0.085) \\ 
  & & & & & \\ 
 Constant & 6.240$^{***}$ & 5.540$^{***}$ & 5.813$^{***}$ & 5.697$^{***}$ & 5.375$^{***}$ \\ 
  & (0.078) & (0.108) & (0.139) & (0.161) & (0.175) \\ 
  & & & & & \\ 
\hline \\[-1.8ex] 
Observations & 6,300 & 6,187 & 5,825 & 5,713 & 6,180 \\ 
R$^{2}$ & 0.042 & 0.046 & 0.025 & 0.019 & 0.018 \\ 
Adjusted R$^{2}$ & 0.041 & 0.044 & 0.022 & 0.016 & 0.014 \\ 
Residual Std. Error & 2.513 (df = 6294) & 2.668 (df = 6177) & 2.548 (df = 5808) & 2.680 (df = 5693) & 2.702 (df = 6154) \\ 
F Statistic & 55.115$^{***}$ (df = 5; 6294) & 32.840$^{***}$ (df = 9; 6177) & 9.211$^{***}$ (df = 16; 5808) & 5.947$^{***}$ (df = 19; 5693) & 4.481$^{***}$ (df = 25; 6154) \\ 
\hline 
\hline \\[-1.8ex] 
\textit{Note:}  & \multicolumn{5}{r}{$^{*}$p$<$0.1; $^{**}$p$<$0.05; $^{***}$p$<$0.01} \\ 
\end{tabular} 
}
\caption{Full Results for Candidate-Choice Conjoint Experiments}\label{tab:full}
\end{table}

\clearpage

\begin{table}[]
\resizebox{15cm}{!}{
\begin{tabular}{@{\extracolsep{5pt}}lccccc} 
\\[-1.8ex]\hline 
\hline \\[-1.8ex] 
 & \multicolumn{5}{c}{Dependent Variable: Preference Score} \\ 
\cline{2-6} 
\\[-1.8ex] & (1) & (2) & (3) & (4) & (5)\\ 
\hline \\[-1.8ex] 
 genderMale & 0.453$^{***}$ & 0.098 & 0.153 & 0.103 & 0.272$^{**}$ \\ 
  & (0.101) & (0.099) & (0.107) & (0.097) & (0.109) \\ 
  & & & & & \\ 
 age52 & 0.088 & 0.028 & $-$0.398$^{**}$ & $-$0.151 & 0.200 \\ 
  & (0.173) & (0.174) & (0.184) & (0.174) & (0.193) \\ 
  & & & & & \\ 
 age60 & $-$0.079 & 0.157 & $-$0.505$^{***}$ & $-$0.001 & 0.347$^{*}$ \\ 
  & (0.173) & (0.173) & (0.184) & (0.174) & (0.193) \\ 
  & & & & & \\ 
 age68 & $-$0.366$^{**}$ & 0.128 & $-$0.575$^{***}$ & $-$0.039 & 0.383$^{**}$ \\ 
  & (0.174) & (0.173) & (0.186) & (0.174) & (0.193) \\ 
  & & & & & \\ 
 age75 & $-$0.855$^{***}$ & $-$0.318 & $-$1.019$^{***}$ & $-$0.385$^{*}$ & $-$0.102 \\ 
  & (0.205) & (0.214) & (0.234) & (0.212) & (0.239) \\ 
  & & & & & \\ 
 eduBA from Ivy League University &  & $-$0.792$^{***}$ & $-$0.427$^{**}$ & $-$0.287$^{*}$ & 0.096 \\ 
  &  & (0.162) & (0.179) & (0.163) & (0.192) \\ 
  & & & & & \\ 
 eduBA from State University &  & $-$0.281$^{**}$ & $-$0.048 & $-$0.268$^{**}$ & 0.132 \\ 
  &  & (0.126) & (0.137) & (0.124) & (0.143) \\ 
  & & & & & \\ 
 taxOpposite &  & $-$0.184 & $-$0.372$^{***}$ & $-$0.483$^{***}$ & $-$0.170 \\ 
  &  & (0.129) & (0.140) & (0.125) & (0.145) \\ 
  & & & & & \\ 
 taxSupport &  & 1.034$^{***}$ & 0.498$^{***}$ & 0.342$^{***}$ & 0.271$^{*}$ \\ 
  &  & (0.126) & (0.138) & (0.125) & (0.143) \\ 
  & & & & & \\ 
 raceWhite &  &  & 0.056 & 0.095 & $-$0.010 \\ 
  &  &  & (0.108) & (0.103) & (0.116) \\ 
  & & & & & \\ 
 income32,000 &  &  & $-$0.045 & $-$0.338$^{*}$ & 0.021 \\ 
  &  &  & (0.211) & (0.186) & (0.214) \\ 
  & & & & & \\ 
 income5.1 million &  &  & $-$0.260 & $-$0.774$^{***}$ & 0.058 \\ 
  &  &  & (0.212) & (0.189) & (0.217) \\ 
  & & & & & \\ 
 income54,000 &  &  & $-$0.238 & $-$0.344$^{**}$ & 0.024 \\ 
  &  &  & (0.167) & (0.150) & (0.171) \\ 
  & & & & & \\ 
 income65,000 &  &  & $-$0.079 & 0.013 & 0.078 \\ 
  &  &  & (0.166) & (0.151) & (0.170) \\ 
  & & & & & \\ 
 income92,000 &  &  & $-$0.047 & $-$0.250$^{*}$ & 0.172 \\ 
  &  &  & (0.165) & (0.151) & (0.171) \\ 
  & & & & & \\ 
 miliServed &  &  &  & 0.088 & 0.069 \\ 
  &  &  &  & (0.101) & (0.114) \\ 
  & & & & & \\ 
 reliProtestant &  &  &  & 0.116 & $-$0.041 \\ 
  &  &  &  & (0.110) & (0.126) \\ 
  & & & & & \\ 
 child1 &  &  &  &  & 0.320$^{*}$ \\ 
  &  &  &  &  & (0.193) \\ 
  & & & & & \\ 
 child2 &  &  &  &  & 0.044 \\ 
  &  &  &  &  & (0.192) \\ 
  & & & & & \\ 
 child3 &  &  &  &  & 0.158 \\ 
  &  &  &  &  & (0.191) \\ 
  & & & & & \\ 
 child4 &  &  &  &  & $-$0.285 \\ 
  &  &  &  &  & (0.239) \\ 
  & & & & & \\ 
 gayOpposite &  &  &  &  & $-$0.453$^{***}$ \\ 
  &  &  &  &  & (0.145) \\ 
  & & & & & \\ 
 gaySupport &  &  &  &  & 0.599$^{***}$ \\ 
  &  &  &  &  & (0.143) \\ 
  & & & & & \\ 
 Constant & 5.707$^{***}$ & 5.664$^{***}$ & 6.406$^{***}$ & 6.071$^{***}$ & 4.829$^{***}$ \\ 
  & (0.154) & (0.180) & (0.222) & (0.234) & (0.300) \\ 
  & & & & & \\ 
\hline \\[-1.8ex] 
Observations & 2,373 & 2,359 & 2,420 & 2,775 & 2,520 \\ 
R$^{2}$ & 0.024 & 0.086 & 0.052 & 0.045 & 0.046 \\ 
Adjusted R$^{2}$ & 0.022 & 0.082 & 0.046 & 0.039 & 0.037 \\ 
Residual Std. Error & 2.456 (df = 2367) & 2.353 (df = 2349) & 2.599 (df = 2404) & 2.493 (df = 2757) & 2.687 (df = 2496) \\ 
F Statistic & 11.462$^{***}$ (df = 5; 2367) & 24.485$^{***}$ (df = 9; 2349) & 8.762$^{***}$ (df = 15; 2404) & 7.611$^{***}$ (df = 17; 2757) & 5.228$^{***}$ (df = 23; 2496) \\ 
\hline 
\hline \\[-1.8ex] 
\textit{Note:}  & \multicolumn{5}{r}{$^{*}$p$<$0.1; $^{**}$p$<$0.05; $^{***}$p$<$0.01} \\ 
\end{tabular} 
}
\caption{Full Results for Candidate-Choice Conjoint Experiments (Control Salience)}\label{tab:full1}
\end{table}

\clearpage

\section{Figures}

\begin{figure}[h]
    \centering
    \includegraphics[scale=0.7]{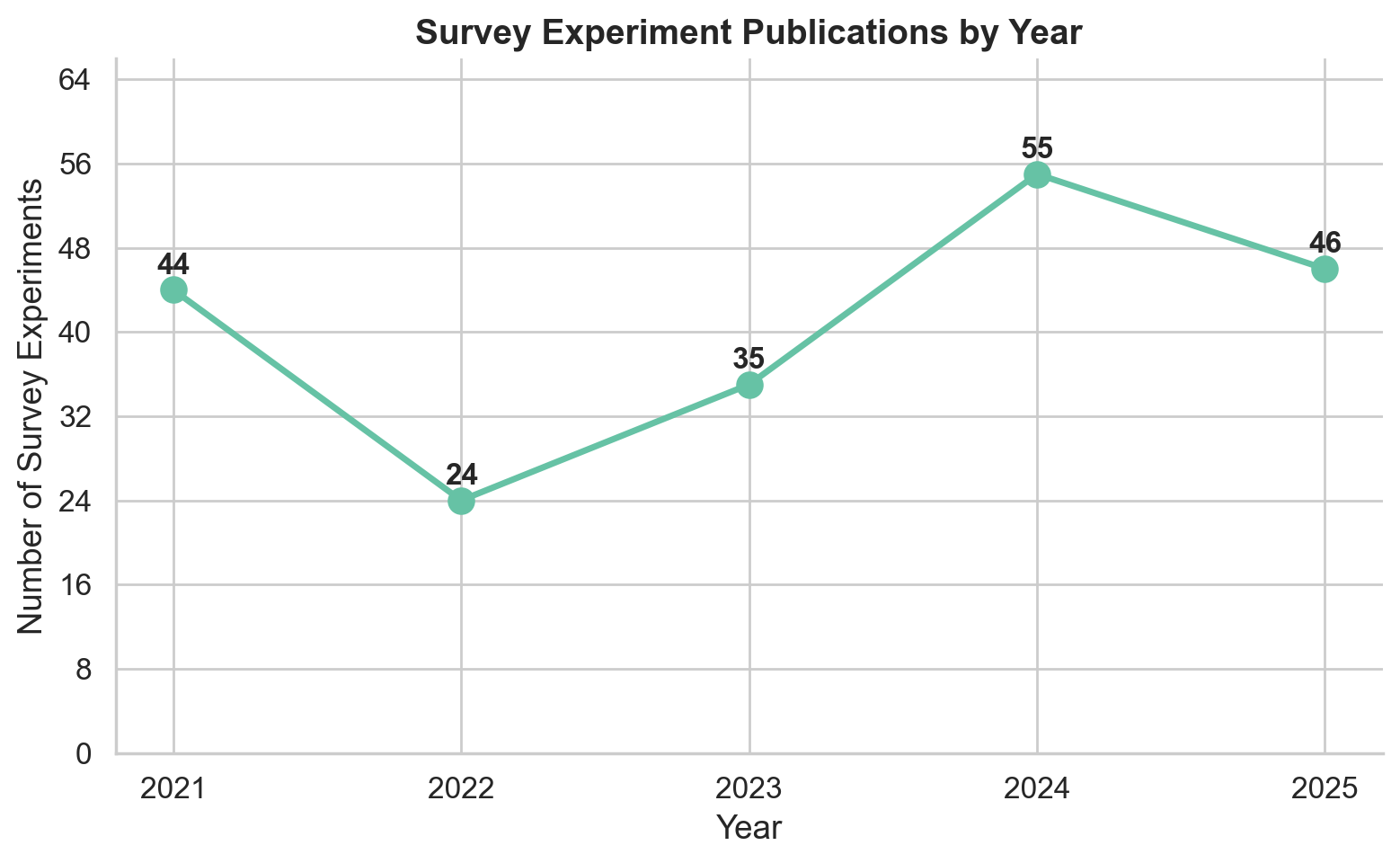}
\end{figure}

\begin{figure}[h]
    \centering
    \includegraphics[scale=0.7]{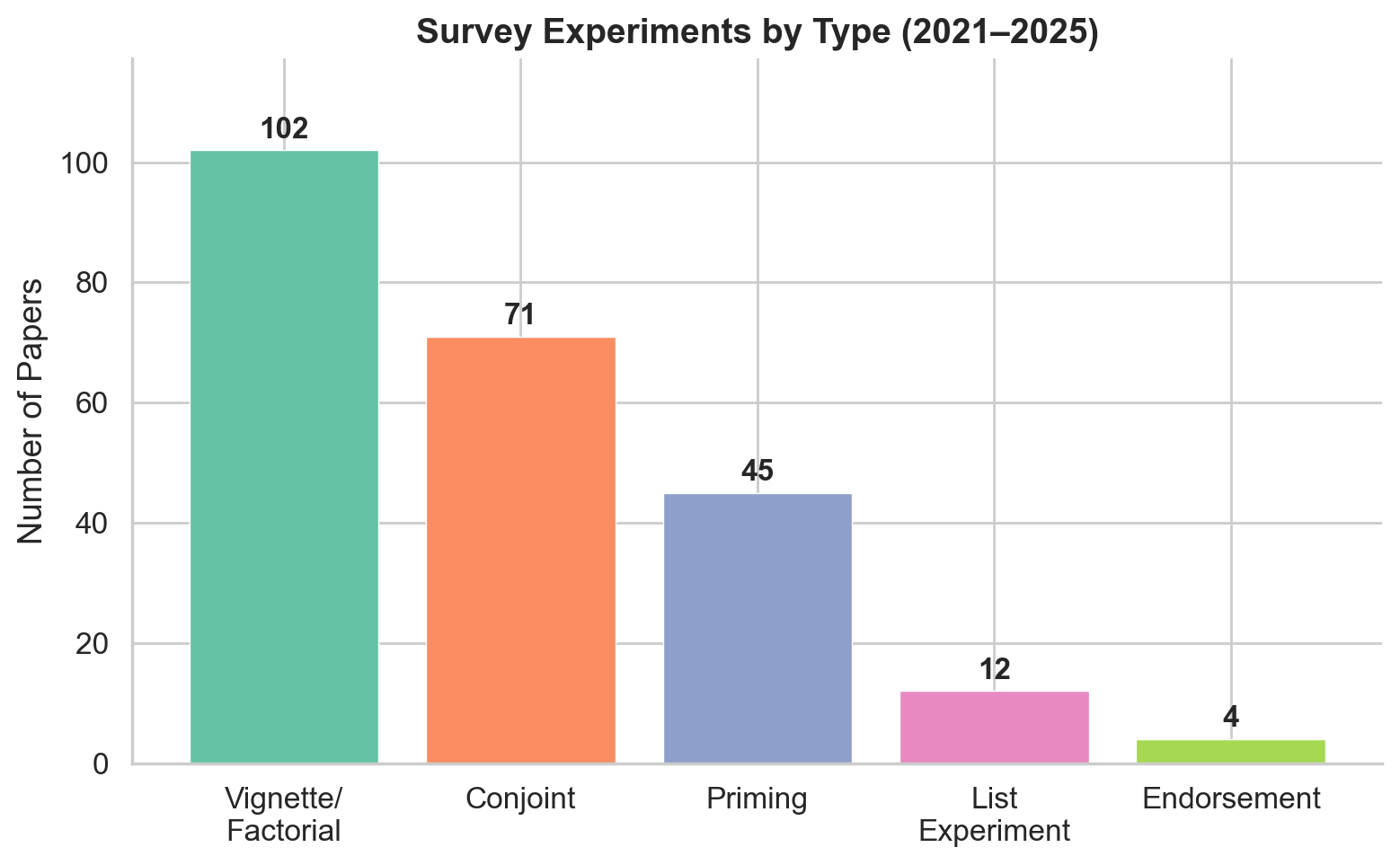}
\end{figure}

\begin{figure}[h]
    \centering
    \includegraphics[scale=0.5]{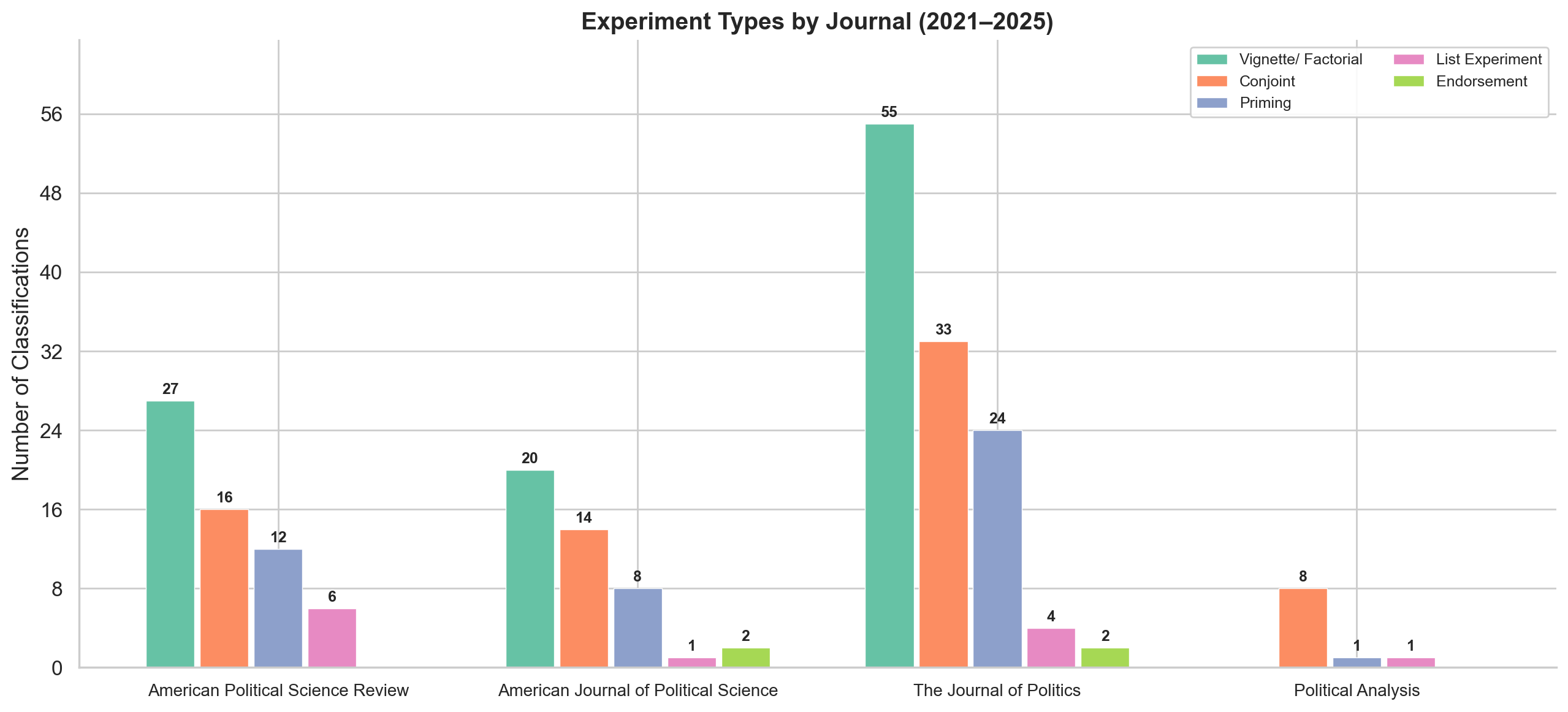}
\end{figure}

\begin{figure}[h]
    \centering
    \includegraphics[scale=0.7]{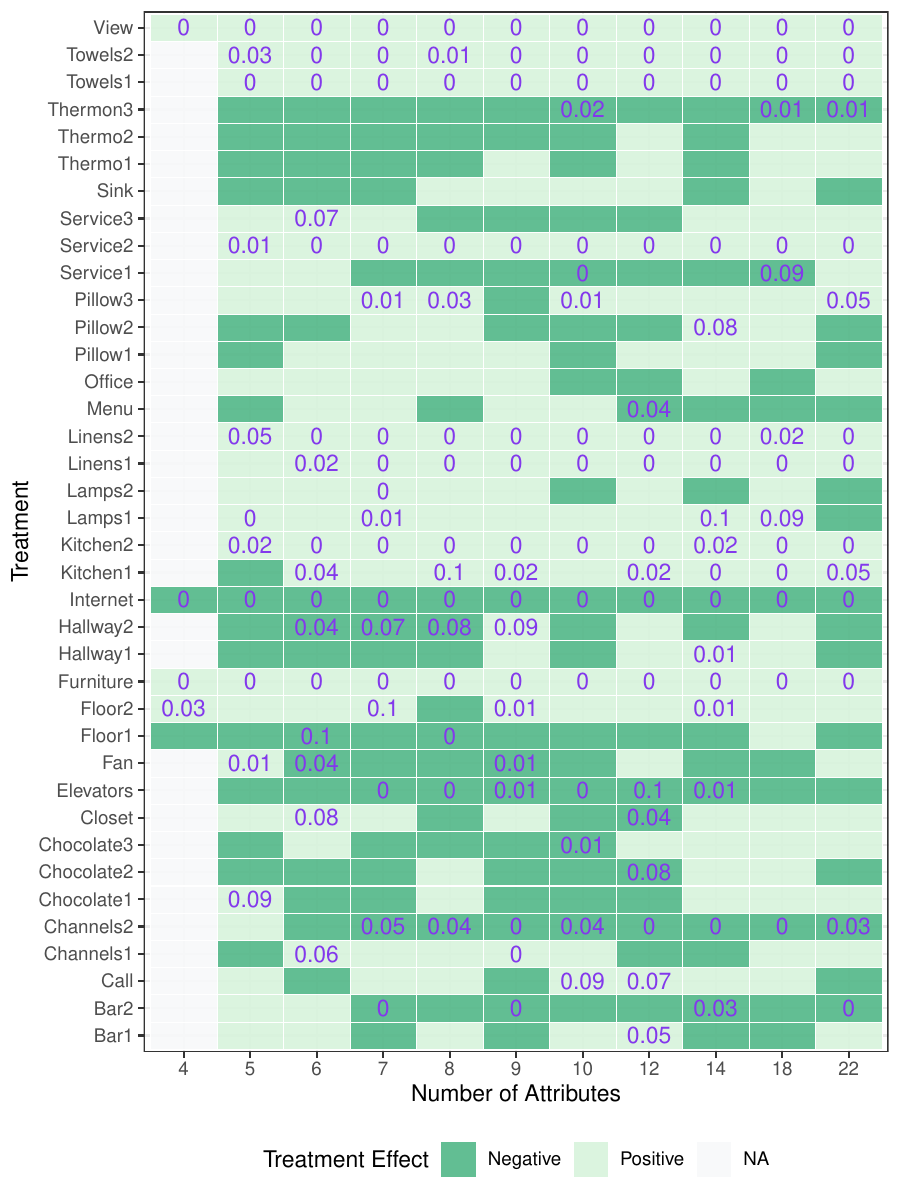}
    \caption{Testing Hypothesis of Effect Reversal with P-value $<0.1$}
    \label{fig:sign_pvalue}
\end{figure}



\end{document}